\documentclass[acmsmall]{acmart}

\def\BibTeX{{\rm B\kern-.05em{\sc i\kern-.025em b}\kern-.08emT\kern-.1667em\lower.7ex\hbox{E}\kern-.125emX}}

\let\OLDthebibliography\thebibliography
\renewcommand\thebibliography[1]{
	\OLDthebibliography{#1}
	\setlength{\parskip}{3pt}
	\setlength{\itemsep}{0pt plus 0.3ex}
}

\usepackage{xspace}
\usepackage{balance}
\usepackage{listings}
\usepackage{color}
\usepackage{cancel}
\usepackage{xcolor}
\usepackage{hyperref}
\usepackage{subfig}
\usepackage{graphicx}
\usepackage{amsmath}
\usepackage{amsthm}
\usepackage{colortbl}
\usepackage[ruled,lined,linesnumbered,vlined,algo2e]{algorithm2e} 
\usepackage{booktabs}
\usepackage{multirow}
\usepackage{enumitem}

\newtheorem{theorem}{Theorem}[section]
\newtheorem{definition}{Definition}[section]

\newcommand{\cf}{\hbox{\emph{cf.}}\xspace}

\newcommand{\etal}{\hbox{\emph{et al.}}\xspace}
\newcommand{\eg}{\hbox{\emph{e.g.}}\xspace}
\newcommand{\ie}{\hbox{\emph{i.e.}}\xspace}
\newcommand{\st}{\hbox{\emph{s.t.}}\xspace}
\newcommand{\wrt}{\hbox{\emph{w.r.t.}}\xspace}
\newcommand{\etc}{\hbox{\emph{etc}}\xspace}
\newcommand{\vs}{\hbox{\emph{vs.}}\xspace}

\definecolor{light-gray}{gray}{0.8}
\newcommand{\highlight}[1]{\colorbox{light-gray}{#1}}

\newcommand{\mynil}{{\mathit{nil}}}

\setcopyright{none}
\renewcommand\footnotetextcopyrightpermission[1]{} 

\begin{document}

%
\title{Towards Efficient Data-flow Test Data Generation}

%

\author{Ting Su}
\affiliation{%
	\institution{Nanyang Technological University}
	\city{Singapore}
	\country{Singapore}}
\email{suting@ntu.edu.sg}

\author{Chengyu Zhang}
\author{Yichen Yan}
\affiliation{%
  \institution{East China Normal University}
  \city{Shanghai}
  \country{China}
}

\author{Lingling Fan}
\affiliation{%
  \institution{Nanyang Technological University}
  \city{Singapore}
  \country{Singapore}}

\author{Geguang Pu}
\affiliation{%
	\institution{East China Normal University}
	\city{Shanghai}
	\country{China}
}

\author{Yang Liu}
\affiliation{%
	\institution{Nanyang Technological University}
	\city{Singapore}
	\country{Singapore}}

\author{Zhendong Su}
\affiliation{%
	\institution{ETH Zurich}
	\city{Zurich}
	\country{Switzerland}
}

%
\renewcommand{\shortauthors}{Ting Su, et al.}

\begin{abstract}
Data-flow testing (DFT) aims to detect potential data interaction anomalies by focusing on the points at which variables receive values and the points at which these values are used. Such test objectives are referred as \emph{def-use pairs}.
However, the complexity of DFT still overwhelms the testers in practice.
To tackle this problem, we introduce a hybrid testing framework for data-flow based test generation: (1) The core of our framework is symbolic execution (SE), enhanced by a novel guided path exploration strategy to improve testing performance; and (2) we systematically cast DFT as reachability checking in software model checking (SMC) to complement SE, yielding practical DFT that combines the two techniques' strengths.
We implemented our framework for C programs on top of the state-of-the-art symbolic execution engine KLEE and instantiated with three different software model checkers.
Our evaluation on the 28,354 def-use pairs collected from 33 open-source and industrial program subjects shows (1) our SE-based approach can improve DFT performance by 15$\sim$48\% in terms of testing time, compared with existing search strategies; and (2) our combined approach can further reduce testing time by 20.1$\sim$93.6\%, and improve data-flow coverage by 27.8$\sim$45.2\% by eliminating infeasible test objectives. 
Compared with the SMC-based approach alone, our combined approach can also reduce testing time by 19.9\%$\sim$23.8\%, and improve data-flow coverage by 7$\sim$10\%.
This combined approach also enables the cross-checking of each component for reliable and robust testing results.
We have made our testing framework and benchmarks publicly available to facilitate future research.
\end{abstract}

\maketitle

\thispagestyle{empty}

\section{Introduction}

It is widely recognized that white-box testing, usually applied at unit testing level, is one of the most important
activities to ensure software quality~\cite{0017273}.
In this process, the testers design inputs to exercise program paths in the code, and validate the outputs with specifications.
Code coverage criteria are popular metrics to guide such test selection.
For example, control-flow based criteria (\eg, statement, branch coverage) require to cover the specified program elements, \eg, statements, branches and conditions, at least once~\cite{Zhu97}.
In contrast, data-flow based criteria~\cite{RappsW85,ClarkePRZ89,HarroldR94} focus on the flow of data, and aim to detect potential data interaction anomalies. It verifies the correctness of variable definitions by observing the values at the corresponding uses.

However, several challenges exist in generating data-flow based test cases: 
(1) \emph{Few data-flow coverage tools exist}. 
To our knowledge, \texttt{ATAC}~\cite{Horgan91,Horgan92} is the only publicly available tool, developed two decades ago, to measure data-flow coverage for \texttt{C} programs.
However, there are plenty of tools for control-flow criteria~\cite{yang2009survey}.
(2) \emph{The complexity of identifying data flow-based test data overwhelms testers}.
Test objectives \wrt data-flow testing are much more than those of control-flow criteria; more efforts are required to satisfy a def-use pair than just covering a statement or branch, since the test case needs to reach a variable definition first and then the corresponding use.
(3) \emph{Infeasible test objectives} (\ie, the paths from the variable definition to the use are infeasible) and \emph{variable aliases} make data-flow testing more difficult. 

To aid data-flow testing, many testing techniques have been proposed in recent years. For example, search-based approach~\cite{Girgis05,GhidukHG07,VivantiMGF13,DenaroMPV15} uses genetic algorithms to guide test generation to cover the target def-use pairs. It generates an initial population of test cases, and iteratively applies mutation and crossover operations on them to optimize the designated fitness function. Random testing~\cite{Girgis05,GirgisS14} generates random test inputs or random paths to cover def-use pairs. Some work uses the idea of collateral coverage~\cite{Malevris05,Santelices07}, \ie, the relation between data-flow criteria and the other criteria (\eg, branch coverage), to infer data-flow based test cases. However, these existing approaches are either \emph{inefficient} (\eg, random testing may generate a large number of redundant test cases) or \emph{imprecise} (\eg, genetic algorithms and collateral coverage-based approach may not be able to identify infeasible test objectives).
 
The aforementioned situations underline the importance of an automated, effective data-flow testing technique, which can efficiently generate test cases for target def-use pairs as well as detect infeasible ones therein.
To this end, we introduce \emph{a combined approach} to automatically generate data-flow based test data. It synergistically combines two techniques, \ie, dynamic symbolic execution and counterexample-guided abstraction refinement-based model checking.
At the high level, given the program under test, our approach (1) \emph{outputs test cases for feasible test objectives},
and (2) \emph{eliminates infeasible test objectives --- without any false positives}.

\underline{D}ynamic \underline{s}ymbolic \underline{e}xecution (DSE)\footnote{Throughout the paper, we will use the term \emph{symbolic execution} to represent the modern symbolic execution technique we adopted in this work without any ambiguity. In fact, the modern symbolic execution techniques include two variants, \ie, \emph{concolic testing} and \emph{execution-generated testing}, which are collectively referred as \emph{dynamic symbolic execution}~\cite{CadarS13}.}~\cite{CadarS13} is a widely-accepted and effective approach for automatic test case generation. It intertwines classic symbolic execution~\cite{symbolicexecution,Clarke76} and concrete execution, and explores as many program paths as possible to generate test cases by solving path constraints.
As for \underline{c}ounter\underline{e}xample-\underline{g}uided \underline{a}bstraction \underline{r}efinement (CEGAR)-based model checking~\cite{BallR02,HenzingerJMS02,ChakiCGJV04,jhala2009software}, given the program source code and a temporal safety specification, it either statically proves that the program satisfies the specification, or returns a counterexample path to demonstrate its violation.
This technique has been used to automatically verify safety properties of OS device drives~\cite{BallR02,Beyer07,BeyerK11}, as well as test generation \wrt statement or branch coverage~\cite{Beyer04,FraserWA09}.

Although symbolic execution has been applied to enforce various coverage criteria (\eg, statement, branch, logical, boundary value and mutation testing)~\cite{Lakhotia09,xietao2010,ZhangXZTHM10,JamrozikFTH12,SuPFHYJZ14},
little effort exists to adapt symbolic execution for data-flow testing.
To counter the path explosion problem, we designed a \emph{cut-point guided path exploration strategy} to cover target def-use pairs as quickly as possible.
The key intuition is to find a set of critical program locations that must
be traversed through in order to cover the pair.
By following these points during the exploration, we can narrow the path search space.
In addition, with the help of path-based exploration, we can also more easily and precisely detect definitions due to variable aliasing.
Moreover, we introduce a simple, powerful encoding of data flow testing using CEGAR-based model checking to complement our SE-based approach:
(1) We show how to encode any data-flow test objective in the program under test and systematically evaluate the technique's practicality; and (2) we describe a combined approach that combines the relative strengths of the SE and CEGAR-based approaches.
An interesting by-product of this combination is to let the two independent approaches cross-check each other's results for correctness and consistency.

In all, this paper makes the following contributions:
\begin{itemize}
	\item It designs a symbolic execution-based testing framework and enhances it with an efficient guided path search strategy to quickly achieve data-flow testing. To our knowledge, our work is the first to adapt symbolic execution for data-flow testing.

	\item It describes a simple, effective reduction of data-flow testing into reachability checking in software model checking to complement our SE-based approach. Again to our knowledge, we are the first to systematically adapt CEGAR-based approach to aid data-flow testing.
	
	\item It implements the SE-based data-flow testing approach, and conducts empirical evaluation on both open-source and industrial C programs. The results show that the SE-based approach is both efficient and effective.
	
	\item It also demonstrates that the CEGAR-based approach can effectively complement the SE-based approach by further reducing testing time and identifying infeasible test objectives. In addition, these two approaches
	can cross-check each other to validate the correctness and effectiveness of both techniques.

\end{itemize}


The original idea of this combined data-flow testing approach was presented in~\cite{SuFPHS15}.
In this article, in addition to providing more technical details and examples, we have made several significant extensions: 
(1) We took substantial efforts to re-implement our SE-based approach on the state-of-the-art symbolic execution engine KLEE~\cite{KLEE}, an execution-generated testing tool (Section~\ref{sec:tool}). 
Previously, the SE-based approach was implemented on our own concolic testing tool CAUT~\cite{WangYSPDH09,SuPFHYJZ14,SuFPHS15}, which was capable of evaluating only 6 program subjects.
Due to the design and architecture differences between KLEE and CAUT, the implementation is not straightforward. 
But our efforts bring several benefits: 
First, it provides a uniform platform to investigate the effectiveness of our exploration strategy with existing ones.
Second, it provides a robust platform to enable extensive evaluation of real-world subjects and better integration with the CEGAR-based approach.
Third, this extension of KLEE could benefit industrial practitioners and also academic researchers to investigate data-flow testing.
Moreover, we optimized the cut-point guided search strategy with several exploration heuristics, and presented the algorithms in this new scenario (Section~\ref{sec:overview} and~\ref{sec:se_approach}).
(2) In addition to the realization of our reduction approach on the CEGAR-based model checking technique, we further investigated the feasibility of this approach on another popular software model checking technique, \ie, bounded model checking (BMC)~\cite{ClarkeKL04}. We extensively evaluated the practicality of both the CEGAR-based approach and BMC-based approach for data-flow testing. Although the BMC-based technique in general cannot eliminate infeasible test objectives as certain, we find it still can serve as (i) a heuristic-criterion to identify hard-to-cover (probably infeasible) test objectives for better prioritizing testing efforts and (ii) a lightweight SMC-based approach for specific types of programs (Section~\ref{sec:study2} and~\ref{sec:study3}).
(3) We further gave the proofs of the correctness of cut-point guided search strategy in symbolic execution and the soundness of our combined approach (Section~\ref{sec:proof}). 
(4) We dedicatedly and rigorously setup a benchmark repository for data-flow testing, and extensively evaluated our approach on 28,354 def-use pairs from 33 program subjects with various data-flow usage scenarios. The benchmarks include 7 non-trivial subjects from previous DFT research work~\cite{Frankl93,MathurW94,HutchinsFGO94,Marre03,GhidukHG07,eler2014covering}, 7 subjects from SIR~\cite{SIR}, 16 subjects from SV-COMP~\cite{SV-COMP}, and 3 industrial projects from our industrial research partners~\cite{PengHSG13,SuPFHYJZ14,MiaoPYSB0CX16,Zhang2018} (Section~\ref{sec:setup}). Based on these subjects, we gave much more detailed evaluation results and analysis (Section~\ref{sec:study1},~\ref{sec:study2} and~\ref{sec:study3}).
In constrast, our original work~\cite{SuFPHS15} only evaluates 6 subjects. 
(5) We included a detailed analysis and discussion on the limitations of our technique, our experience of applying data-flow testing, and the applications for other testing scenarios (Section~\ref{sec:discussion}).
(6) To enable the replication of our results and benefit future research on data-flow testing, we have made all the artifacts (benchmarks, tools, scripts) publicly available at~\cite{dft-artifact}.


The paper is organized as follows. 
Section~\ref{sec:preliminary} defines the problem of data-flow testing and gives necessary background.
Section~\ref{sec:overview} gives an example to motivate and illustrate our testing approach.
Section~\ref{sec:approach} details our approach and algorithms.
Section~\ref{sec:proof} proves the correctness and soundness of our approach.
Section~\ref{sec:tool} gives the details of our tool design and implementation.
Section~\ref{sec:evaluation} presents the evaluation results and analysis.
Section~\ref{sec:discussion} gives a detailed discussion of the limitations of our technique, our experience of applying data-flow testing and the applications for other testing scenarios.
Section~\ref{sec:related} surveys the related work.
Section~\ref{sec:conclusion} concludes the paper.
\section{Problem Definition, Preliminaries and Challenges}
\label{sec:preliminary}

\subsection{Problem Definition}

\begin{definition}[Program Paths]
\label{def:program_path} 
Two kinds of program paths, i.e., control flow paths and execution paths
are distinguished during data-flow testing. 
Control flow paths are the paths from the control flow graph of the program under test,
which abstract the flow of control.
Execution paths are driven by concrete program inputs, which 
represent dynamic program executions.
Both of them can be represented as a sequence of control points (denoted by line numbers),
e.g., $l_1, \ldots, l_i, \ldots, l_n$.
\end{definition}

\begin{definition}[Def-use Pair]
\label{def:dua} 
The test objective of data-flow testing is referred as a def-use pair, denoted by $du(l_d, l_u, v)$. 
Such a pair appears when there exists a control flow path that starts from 
the variable definition statement $l_d$ (or the def statement in short), 
and then reaches the variable use statement $l_u$ (or the use statement in short),
but no statements on the subpaths from $l_d$ to $l_u$ redefine the variable $v$.
\end{definition}

In particular, two kinds of def-use pairs are distinguished. For a def-use pair $(l_d, l_u, v)$, if the variable $v$ is used in a computation statement at $l_u$, the pair is a \emph{computation-use} (\emph{c-use} for short), denoted by $dcu(l_d, l_u, v)$. If $v$ is used in a conditional statement (\eg, an \texttt{if} or \texttt{while} statement) at $l_u$, the pair is a \emph{predicate use} (\emph{p-use} for short). At this time, two def-use pairs appear and can be denoted by $dpu(l_d, (l_u, l_t), v)$ and $dpu(l_d, (l_u, l_f), v)$, where $(l_u, l_t)$ and $(l_u, l_f)$ represents the \emph{true} and the \emph{false} edge of the conditional statement, respectively.

\begin{definition}[Data-flow Testing]
Given a def-use pair $du(l_d, l_u, v)$ in program $P$ under test, the goal of data-flow testing\footnote{In this paper, we focus on the problem of \emph{classic data-flow testing}~\cite{GhidukHG07,VivantiMGF13}, \ie, finding an input for a given def-use pair at one time. We do not consider the case where some pairs can be accidentally covered when targeting one pair, since this has already been investigated in collateral coverage-based approach~\cite{Marre96,Marre03}. We will discuss more in Section~\ref{sec:discussion}.} is to find an input $t$ that induces an execution path $p$ that covers the variable definition statement at $l_d$,
and then covers variable use statement at $l_u$, but without covering any redefinition statements w.r.t 
$v$, i.e., the subpath from $l_d$ to $l_u$ is a def-clear path.
The requirement to cover all def-use pairs at least once is called all def-use coverage criterion\footnote{We follow the all def-use coverage defined by Rapps and Weyuker~\cite{Rapps82,RappsW85}, since almost all of the literature that followed uses or extends this definition, as revealed by a recent survey~\cite{SuWMPHCS17}.} in data-flow testing.
\label{defn:dft} 
\end{definition}

In particular, for a c-use pair, $t$ should cover $l_d$ and $l_u$; for a p-use pair, $t$ should cover $l_d$ and its true or false branch, \ie, $(l_u, l_t)$ and $(l_u, l_f)$, respectively.

\subsection{Symbolic Execution}
\label{sec:se_process}
Our data-flow testing approach is mainly built on the symbolic execution technique.
The idea of symbolic execution (SE) was initially described in~\cite{symbolicexecution,Clarke76}. Recent significant advances in the constraint solving techniques have made SE possible for testing real-world program by systematically exploring program paths~\cite{CadarS13} .
Specifically, two variants of modern SE techniques exist, \ie, \emph{concolic testing} (implemented by DART~\cite{DART}, CUTE~\cite{CUTEC}, CREST~\cite{CREST}, CAUT~\cite{SuPFHYJZ14}, \etc) and \emph{execution-generated testing} (implemented by EXE~\cite{CadarGPDE06} and KLEE~\cite{KLEE}), which mix concrete and symbolic execution together to improve scalability.
In essence, SE uses \emph{symbolic values} in place of \emph{concrete values} to represent input variables, and represent other program variables by the \emph{symbolic expressions} in terms of symbolic inputs.
Typically, SE maintains a \emph{symbolic state} $\sigma$, which maps variables to (1) the symbolic expressions over program variables, and (2) a symbolic path constraint $pc$ (a quantifier-free first order formula in terms of input variables), which characterizes the set of input values that can execute a specific program execution path $p$.
Additionally, $\sigma$ maintains a program counter that refers to the current instruction for execution.
At the beginning, $\sigma$ is initialized as an empty map and $pc$ as \emph{true}.
During execution, SE updates $\sigma$ when an assignment statement is executed; and forks $\sigma$ when a conditional statement (\eg, \textbf{if}($e$) $s_1$ \textbf{else} $s_2$) is executed. Specifically, SE will create a new state $\sigma'$ from the original state $\sigma$, and updates the path constrain of $\sigma'$ as $pc \wedge \neg(e)$, while updates that of $\sigma$ as $pc \wedge (e)$. $\sigma$ and $\sigma'$, respectively, represent the two program states that fork at the \emph{true} and \emph{false} branch of the conditional statement. By querying the satisfiability of updated path constraints, SE decides which one to continue the exploration. When an \emph{exit} or certain runtime error is encountered, SE will terminate on that statement and the concrete input values will be generated by solving the corresponding path constraint. 

\vspace*{3pt}
\noindent\emph{\textbf{Challenges}}. \quad
Although SE is an effective test case generation technique for traditional coverage criteria, it faces two challenges when applied in data-flow testing:
\begin{enumerate}
	\item The SE-based approach by nature faces the notorious path-explosion problem. Despite the existence of many generic search strategies, it is challenging, in reasonable time, to find an execution path from the whole path space to cover a given pair.
	\item The test objectives from data-flow testing include \emph{feasible} and \emph{infeasible} pairs. 
	A pair is \emph{feasible} if there exists an execution path which can pass through it. 
	Otherwise it is \emph{infeasible}.
	Without prior knowledge about whether a target pair is feasible or not, the SE-based approach may spend a large amount of time, in vain, to cover an infeasible def-use pair. 
\end{enumerate}

Section~\ref{sec:overview} will give an overview of our approach, and illustrate how our combined approach tackles these two challenges via an example in Fig.~\ref{fig:example}.

\begin{figure}
\begin{center}
	\includegraphics[width=0.6\textwidth]{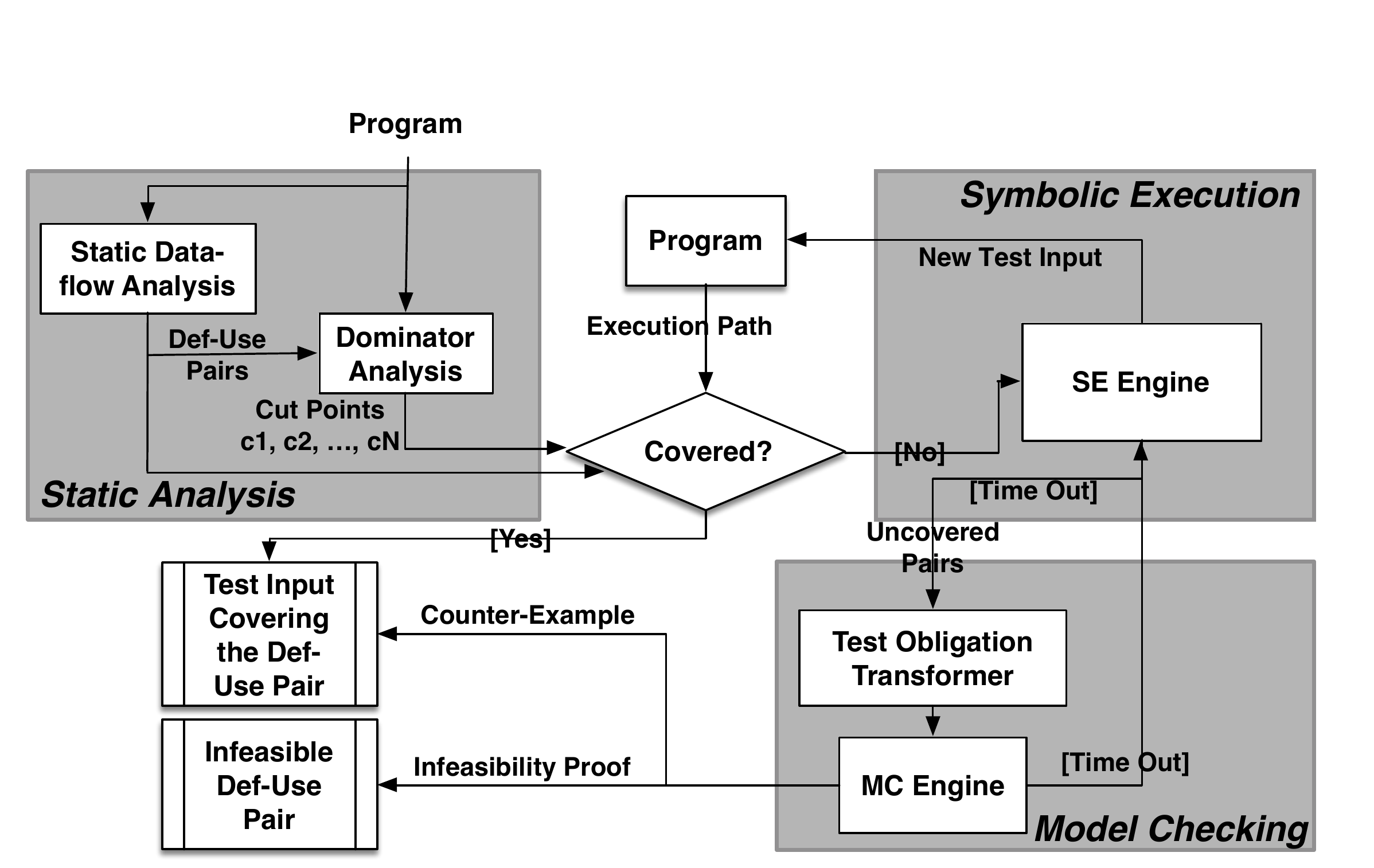}
\end{center}
\caption{Workflow of the combined approach for data-flow testing, which combines symbolic execution and software model checking (the CEGAR-based model checking in particular).}
\label{fig:dft_workflow}
\end{figure}

\section{Approach Overview}
\label{sec:overview}

Fig.~\ref{fig:dft_workflow} shows the workflow of our combined approach for data-flow testing.
It takes as input the program source code, and follows the three steps below to achieve automated, efficient DFT.
(1) The \emph{static analysis} module uses data-flow analysis to identify def-use pairs, and adopts dominator analysis to analyze the sequence of cut points for each pair (see Section~\ref{sec:static_analysis}).
(2) For each pair, the \emph{symbolic execution} module adopts the cut point-guided search strategy to efficiently find an execution path that could cover it within a specified time bound (see Section~\ref{sec:se_approach}).
(3) For the remaining uncovered (possibly infeasible) pairs, the \emph{software model checking} module encodes the test obligation of each def-use pair into the program under test, and enforces reachability checking (also within a time bound) on each of them. The model checker can eliminate infeasible ones with proofs and may also identify feasible ones (see Section~\ref{sec:mc_approach}). 
If the testing resource permits, the framework can iterate between (2) and (3) by lifting the time bound to continue test those remaining uncovered pairs. By this way, our framework outputs test cases for feasible test objectives, and weeds out infeasible ones by proofs --- without any false positives.

\vspace{3pt}
\subsection{Illustrative Example} 
Fig.~\ref{fig:example} shows an example program \emph{power}, 
which accepts two integers $x$ and $y$, and outputs the result of $x^y$.
The right sub-figure shows the control flow graph of \emph{power}.

\vspace{3pt}
\noindent{\emph{\textbf{Step 1: Static Analysis}}}. 
For the variable $res$ (it stores the computation result of $x^y$), the static analysis procedure can find two typical def-use pairs with their cut points:
\begin{align}
du_1&=(l_8, l_{17}, res)\\
du_2&=(l_8, l_{18}, res)
\end{align}

Below, we illustrate how our combined approach can efficiently achieve DFT on
these two def-use pairs ---
SE can efficiently cover the feasible pair $du_1$, and 
CEGAR can effectively conclude the infeasibility of $du_2$.

\lstset{ %
    language=C++,                
    basicstyle=\footnotesize,       
    numbers=left,                   
    numberstyle=\ttfamily\scriptsize,      
    stepnumber=1,                   
    numbersep=5pt,                  
    backgroundcolor=\color{white},  
    showspaces=false,               
    showstringspaces=false,         
    showtabs=false,                 
    frame=false,           
    tabsize=2,          
    captionpos=b,           
    xleftmargin=2em,xrightmargin=2em, aboveskip=1em,
    breaklines=true,        
    breakatwhitespace=false,    
   escapeinside={(*@}{@*)},
}

\lstset{language=C, numbers=left, xleftmargin=2em, xrightmargin=-2em,  numberstyle=\scriptsize, 
 escapeinside={@}{@}, basicstyle=\footnotesize\ttfamily, numbersep=8pt, commentstyle=\itshape\scriptsize}

\begin{figure}
  \begin{minipage}[b]{.4\linewidth}
\begin{lstlisting}
double power(int x,int y){
    int exp; @\label{l:entry}@
    double res;
    if (y>0) @\label{l:use1_y}@
      exp = y;@\label{l:def1_exp_use2_y}@
    else
      exp = -y;@\label{l:def2_exp_use3_y}@
    res=1; @\label{l:def1_res}@
    while (exp!=0){ @\label{l:use1_exp}@
      res *= x; @\label{l:def2_res_use1_res_use1_x}@
      exp -= 1; @\label{l:def3_exp_use2_exp}@
    }
    if (y<=0) @\label{l:use3_y}@
      if(x==0)
        @{\bf{abort}}@;
      else 
        return 1.0/res; @\label{l:def3_res_use2_res}@
    return res; @\label{l:use3_res}@
}\end{lstlisting}%
  \end{minipage}
  \begin{minipage}[b]{.3\linewidth}
\includegraphics[width=\textwidth]{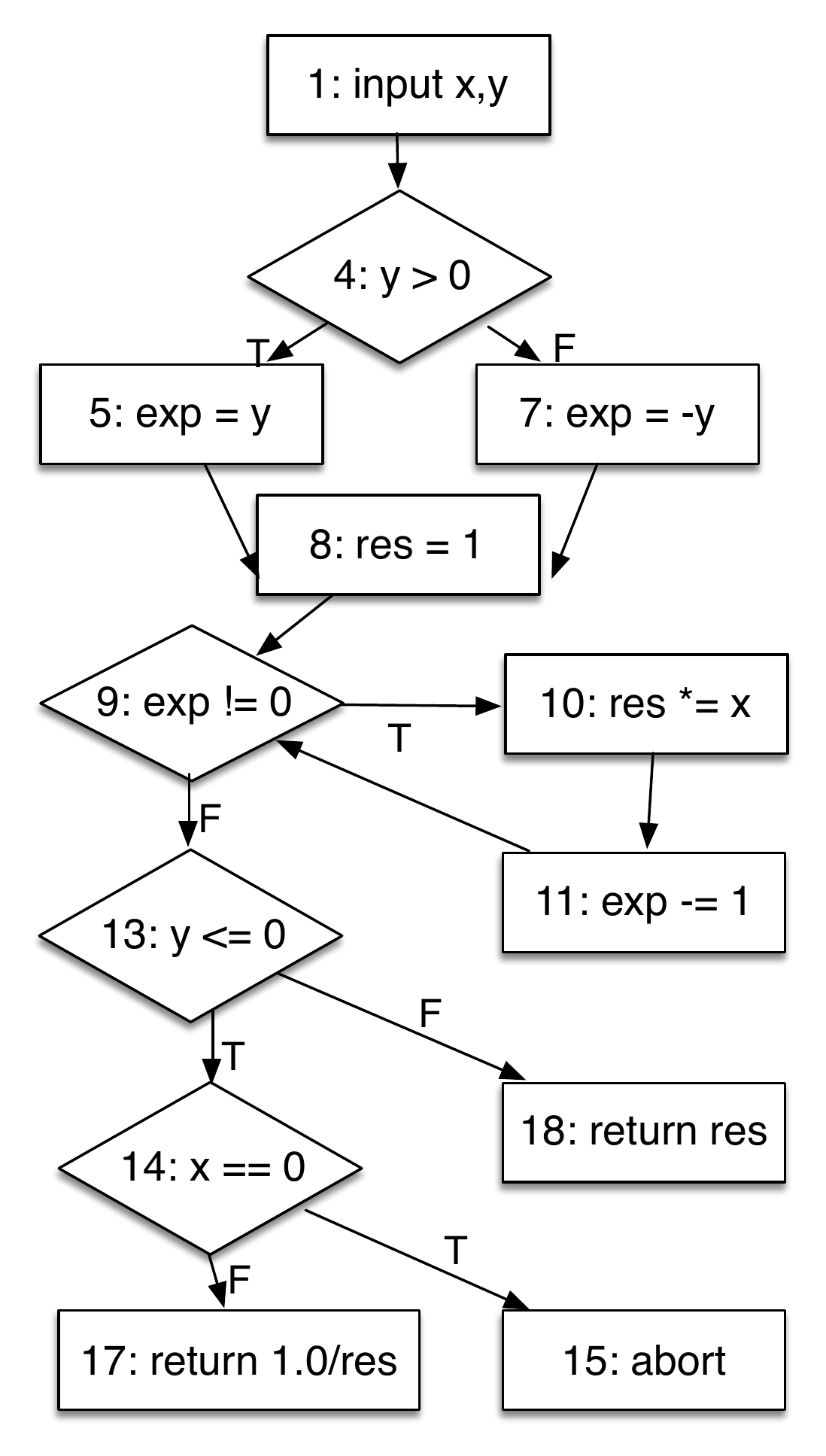}
\label{fig:example_right}
  \end{minipage}
  \caption{An example: \emph{power}.}
\label{fig:example}
\end{figure}

\begin{table*}[!t]
\footnotesize
\newcommand{\tabincell}[2]{\begin{tabular}{@{}#1@{}}#2\end{tabular}}
\renewcommand{\arraystretch}{1}
\caption{Running steps of the enhanced symbolic execution approach for data-flow testing.}
\label{table:table_enhanced_se}
\centering
\resizebox{\linewidth}{!}{
\begin{tabular}{|c||c|c|c|c|}
\hline
Steps &Pending Path &Priority Queue &Selected Path &Path Constraint (\emph{pc})\\
\hline

1
& \textbf{1}: $l_{4T}$, \textbf{2}: $l_{4F}$
& $(l_4,2)^{\textbf{1}}$, $(l_4,2)^{\textbf{2}}$
& \textbf{1}
& $y>0$\\
\hline

2
& \tabincell{c}{\textbf{2}: $l_{4F}$, \\\textbf{3}: $l_{4T}, l_{9T}$, \textbf{4}: $l_{4T}, l_{9F}$}
& \tabincell{c}{ $(l_9,1)^\textbf{4}$, \\$(l_4,2)^\textbf{2}$, $(l_9,4)^\textbf{3}$}
& \textbf{4}
& $y>0 \wedge y==0$ \\
\hline

3
&  \tabincell{c}{\textbf{2}: $l_{4F}$, \textbf{3}: $l_{4T}, l_{9T}$}
&  \tabincell{c}{$(l_4,2)^\textbf{2}$, $(l_9,4)^\textbf{3}$}
& \textbf{3}
& $y>0 \wedge y!=0$ \\
\hline

4
& \tabincell{c}{\textbf{2}: $l_{4F}$, \\\textbf{5}: $l_{4T}, l_{9T}, l_{9T}$, \textbf{6}: $l_{4T}, l_{9T}, l_{9F}$ }
& \tabincell{c}{$(l_4,2)^\textbf{2}$, \\$(l_9,1)^\textbf{5}$, $(l_9,1)^\textbf{6}$}
& \tabincell{c}{prune \textbf{5,6}, \\select \textbf{2}}
& $y\leq0$\\
\hline

5
& \textbf{7}: $l_{4F}, l_{9T}$, \textbf{8}: $l_{4F}, l_{9F}$
& $(l_9,4)^\textbf{7}$, $(l_9,1)^\textbf{8}$
& \textbf{8}
& $y<0 \wedge y!=0$\\
\hline

6
& \tabincell{c}{\textbf{7}: $l_{4F}, l_{9T}$, \\\textbf{9}: $l_{4F}, l_{9F}, l_{13T}$, \textbf{10}: $l_{4F}, l_{9F}, l_{13F}$}
& \tabincell{c}{$(l_{13},1)^\textbf{9}$, \\$(l_9,4)^\textbf{7}$, $(l_{13},\infty)^\textbf{10}$}
& \textbf{9}
& $y\leq0 \wedge y!=0$\\
\hline

7
& \tabincell{c}{\textbf{7}: $l_{4F}, l_{9T}$, \textbf{10}: $l_{4F}, l_{9F}, l_{13F}$ \\\textbf{11}: $l_{4F}, l_{9F}, l_{13T}, l_{14T}$, \textbf{12}: $l_{4F}, l_{9F}, l_{13T}, l_{14F}$}
&\tabincell{c}{$(l_{13},1)^\textbf{12}$, \\$(l_9,4)^\textbf{7}$, $(l_{13},\infty)^\textbf{10}$, $(l_{14},\infty)^\textbf{11}$}
& \textbf{12}
& $y==0 \wedge x!=0$\\
\hline

\end{tabular}}
\end{table*}

\vspace*{3pt}
\noindent{\emph{\textbf{Step 2: SE-based Data-flow Testing}}}. \quad
When SE is used to cover $du_1$, assume under the classic depth-first search (DFS) strategy~\cite{CUTEC,DART,KLEE,PEX,CREST,SuPFHYJZ14} the \emph{true} branches of the new execution states are always first selected, we can get an execution path $p$ after unfolding the \texttt{while} loops $n$ times.
\begin{equation}\label{eq:dfs_path}
 p=l_4, l_5, l_8, \underbrace{l_9, l_{10}, l_{11},\mbox{  } l_9,l_{10},l_{11}, \ldots }_{\text{repeated}~ n ~ \text{times}}, l_9, l_{13}, l_{14}, l_{15}
\end{equation}
Here $p$ already covers the definition statement (at $l_8$) \wrt the variable $res$.
In order to cover the use statement (at $l_{17}$), SE will exhaustively execute program
paths by exploring the remaining unexecuted branch directions.
However, the \emph{path (state) explosion problem} --- hundreds of branch directions exist (including
those branches from the new explored paths) --- will drastically slow down
data-flow testing.

To mitigate this problem, the key idea of our approach is to reduce unnecessary path exploration and provide more guidance during execution.
To achieve this, we designed a novel \emph{cut-point guided search algorithm} (CPGS) to enhance SE, which leverages
several key elements to prioritize the selection of ESs:
(1) \emph{a guided search algorithm}, which leverages two metrics: (i) \emph{cut points}, a sequence of control points that must be traversed through for any paths that could cover the target pair.
For example, the cut points of $du_1$ are \{$l_4$, $l_8$, $l_9$, $l_{13}$, $l_{14}$, $l_{17}$\}. These critical points are used as intermediate goals during the search to narrow down the exploration space of SE. 
(ii) \emph{instruction distance}, the distance between an ES and a target search goal in terms of number of program instructions on the control flow graph. Intuitively, an ES with closer (instruction) distance toward the goal can reach it more quickly. For example, when SE reaches $l_9$, it can fork two execution states, \ie, following the \emph{true} and the \emph{false} branches. If our target goal is to reach $l_{13}$, the \emph{false} branch will be prioritized since it has 1-instruction distance toward $l_{13}$, while the opposite branch has 3-instruction distance.
More specially, CPGS is enhanced with (2) \emph{a backtrack strategy} based on the number of executed instructions, which reduces the likelihood of trapping in tight loops; and (3) \emph{a redefinition path pruning technique}, which detects and removes redundant ESs.

Table~\ref{table:table_enhanced_se} shows the steps taken by our cut-point guided search algorithm to cover $du_1$.
At the beginning, SE forks two ESs for the \texttt{if} statement at $l_4$, which produces two \emph{pending paths}\footnote{An pending path indicates a not fully-explored path (corresponding to an unterminated state).}, \ie, $l_{4T}$ and $l_{4F}$\footnote{We use the line number followed by $T$ or $F$ to denote the \emph{true} or \emph{false} branch of the \texttt{if} statement at the corresponding line.}. In detail, we maintain a tuple $(c,d)^\textbf{i}$ that records the two aforementioned metrics for each pending path $\textbf{i}$ in a priority queue, where $c$ is the deepest covered cut point, and $d$ is the shortest distance between the corresponding ES and the next target cut point. In each step, we choose the pending path $\textbf{i}$ with the optimal value $(c,d)$. For example, in Step \textbf{1}, Path \textbf{1} and Path \textbf{2} have the same values $(l_4, 2)$, and thus we randomly select one path, \eg, Path \textbf{1}. 

Later, in Step \textbf{2}, Path \textbf{1} produces two new pending paths, Path \textbf{3} and Path \textbf{4}. We choose Path \textbf{4} since it has the best value: it has sequentially covered the cut points \{$l_4$, $l_8$, $l_9$\}, and it is closer to the next cut point $l_{13}$ than Path \textbf{3} on the control flow graph, so it is more likely to reach $l_{13}$ more quickly. However, its $pc$ is unsatisfiable. As a result, we give up exploring this pending path, and choose Path \textbf{3} (because it covers more cut points than Path \textbf{1}) in the next Step \textbf{3}, which induces Path \textbf{5} and Path \textbf{6}. At this time, our algorithm detects the variable $res$ is redefined at $l_{10}$ on Path \textbf{5} and Path \textbf{6}, according to the definition of DFT, it is useless to explore these two paths. So, Path \textbf{5} and Path \textbf{6} are pruned. This \emph{redefinition path pruning} technique can rule out these invalid paths to speed up DFT. Note despite only two pending paths are removed in this case, a number of potential paths have actually been prevented from execution (see the example path in (\ref{eq:dfs_path})), which can largely improve the performance of our search algorithm. 

We choose the only remaining Path \textbf{2} to continue the exploration, which produces Path \textbf{7} and Path \textbf{8} in Step \textbf{5}. Again, we choose Path \textbf{8} to explore, which induces Path \textbf{9} and \textbf{10} in Step \textbf{6}. Here, for Path \textbf{10}, due to it cannot reach the next target point $l_{14}$, its distance is set as $\infty$. As last, Path \textbf{9} is selected, and our algorithm finds Path \textbf{12} which covers $du_1$, and by solving its path constraint $y==0\wedge x!=0$, we can get one test input, \eg, $t=(x\mapsto 1, y\mapsto 0)$, to satisfy the pair. The above process is enforced by the cut-point guided search, which only takes 7 steps to cover $du_1$.
As we will demonstrate in Section~\ref{sec:evaluation}, the cut point-guided search strategy is more effective for data-flow testing than the existing state-of-the-art search algorithms.

\vspace*{3pt}
\noindent{\emph{\textbf{Step 3: CEGAR-based Data-flow Testing}}}. \quad
In data-flow testing, classic data-flow analysis techniques~\cite{Harrold94,Pande94,Chatterjee99} statically identify def-use pairs by analyzing data-flow relations.
However, due to its conservativeness and limitations, infeasible pairs may be included,
which greatly affects the effectiveness of SE for DFT.
For example, the pair $du_2$ is identified as a def-use pair since there exists a def-clear control-flow path (\ie, $l_8, l_9, l_{13}, l_{18}$) that can start from the variable definition (\ie, $l_8$) and reach the use (\ie, $l_{18}$). However,
$du_2$ is infeasible (\ie, no test inputs can satisfy it):
If we want to cover its use statement at $l_{18}$, we cannot take the true branch
of $l_{13}$, so $y>0$ should hold.
However, if $y>0$, the variable $exp$ will be assigned a positive value at $l_5$ by
taking the true branch of $l_4$, and 
the redefinition statement at $l_{10}$ \wrt the variable $res$ will be executed.
As a result, such a path that covers the pair and avoids the redefinition at the same time
does not exist, and $du_2$ is an infeasible pair.
It is rather difficult for SE to conclude the feasibility unless it checks all program paths, which however is almost impossible due to infinite paths in real-world programs.

\begin{figure}
\begin{tabular}{c}
\begin{lstlisting}
double power(int x, int y){
  @{\highlight{bool cover\_flag = false;}}@
  int exp; @\label{l:entry}@
  double res;
  ...
  res=1; @\label{l:def1_res}@
  @{\highlight{cover\_flag = true;}}@
  while (exp!=0){ @\label{l:use1_exp}@
     res *= x; @\label{l:def2_res_use1_res_use1_x}@
     @{\highlight{cover\_flag = false;}}@
     exp -= 1; @\label{l:def3_exp_use2_exp}@
   }
   ...
  @{\highlight{if(cover\_flag){ check\_point(); }}}@
   return res; @\label{l:use3_res}@
}
\end{lstlisting}%
\end{tabular} 
  \caption{The transformed function \emph{power} with the test requirement encoded in highlighted statements.}
\label{fig:power2}
\end{figure}

To counter the problem, our key idea is to reduce the data-flow testing problem into the
path reachability checking problem in software model checking.
We encode the test obligation of a target def-use pair into the program under test, and
leverage the power of model checkers to check its feasibility.
For example, in order to check the feasibility of $du_2$, we instrument the test requirement
into the program as shown in Fig.~\ref{fig:power2}.
We first introduce a boolean variable \emph{cover\_flag} at $l_2$, and initialize it as \texttt{false},
which represents the coverage status of this pair.
After the definition statement, the variable \emph{cover\_flag} is set as \texttt{true} (at $l_7$);
\emph{cover\_flag} is set as \texttt{false} immediately after all the redefinition statements (at $l_{10}$).
We check whether the property \emph{cover\_flag==true} holds (at $l_{14}$) 
just before the use statement.
If the check point is reachable, the pair is feasible and a test case will be generated.
Otherwise, the pair is infeasible, and will be excluded in the coverage computation.
As we can see, this model checking based approach is flexible and can be fully automated.

\vspace*{3pt}
\noindent{\emph{\textbf{Combined SE-CEGAR based Data-flow Testing}}}. \quad
In data-flow testing, the set of test objectives include feasible and infeasible pairs.
As we can see from the above two examples, SE, as a dynamic path-based exploration approach, can efficiently cover feasible pairs; 
while CEGAR, as a static software model checking approach, can effectively detect infeasible pairs (may also cover some feasible pairs).

\hspace*{4.5cm}\includegraphics[scale=0.5]{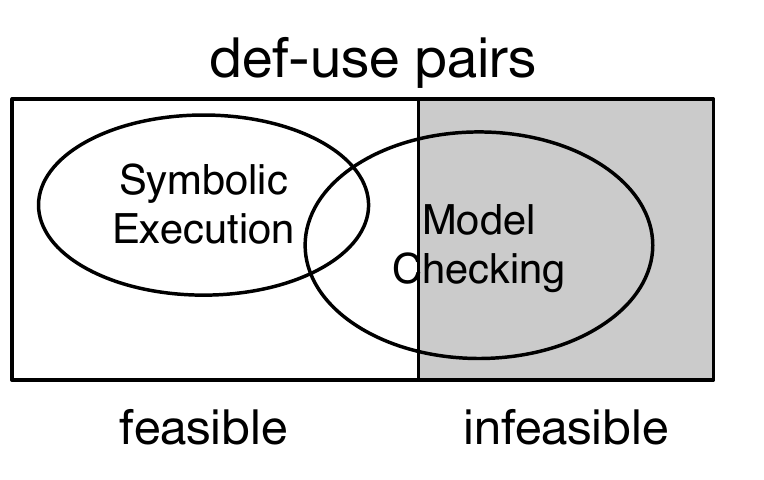}

The figure above  shows the relation of these two approaches for data-flow testing.
The white part represents the set of feasible pairs, and the gray part the set of infeasible ones.
The SE-based approach is able to cover feasible pairs efficiently, but in general, due to the path explosion problem, it cannot detect infeasible pairs (this may waste a lot of testing time). 
The CEGAR-based approach is able to identify infeasible pairs efficiently (but may take more time to cover feasible ones).
As a result, it is beneficial to combine these two techniques to complement each other with their strengths. 
Section~\ref{sec:evaluation} will demonstrate our observations, and validate the combined approach can indeed achieve more efficient data-flow testing by reducing testing time as well as improving coverage, compared to either the SE-based approach or the CEGAR-based approach alone.

\section{Our Approach}
\label{sec:approach}
This section explains the details of our combined approach for data-flow testing.
It consists of three phases: (1) use static analysis to identify
data flow-based test objectives and collect necessary program information; (2)
use symbolic execution and (3) software model checking to generate test cases for feasible test objectives
and eliminate infeasible ones at the same time.

\subsection{Static Analysis}
\label{sec:static_analysis}
We use standard iterative data-flow analysis~\cite{Harrold94,Pande94}
to identify def-use pairs from the program under test (see Section~\ref{sec:tool} for implementation details). 
To improve the performance of SE-based data-flow testing, we use dominator analysis to analyze a set of \emph{cut points} to
effectively guide path exploration. In the following, we give some definitions.

\begin{definition}[Dominator]
\label{def:dominator} 
In a control-flow graph, a node $m$ dominates a node $n$ if all paths from the program entry to
$n$ must go through $m$, which is denoted as $m \gg n$. 
When $m \neq n$, we say $m$ strictly dominates $n$.
If $m$ is the unique node that strictly dominates $n$ and does not strictly dominate other nodes
that strictly dominate $n$, $m$ is an immediate dominator of $n$, 
denoted as $m \gg^{I} n$.
\end{definition}

\begin{definition}[Cut Point]
Given a def-use pair $du(l_d, l_u, v)$, its cut points are a sequence of critical control points 
$c_1, \ldots, c_i, \ldots, c_n$ that must be passed through in succession by any control flow paths 
that cover this pair.
The latter control point is the immediate dominator of the former one, \ie,
$c_1 \gg^{I}  \ldots c_i \gg^{I} l_d \gg^{I} \ldots c_n \gg^{I}  l_u$.
Each control point in this sequence is called a cut point.
\label{def:cutpoint} 
\end{definition}

Note the \emph{def} and the \emph{use} statement (\ie, $l_d$ and $l_u$) 
of the pair itself also serve as the cut points.
These cut points are used as the intermediate goals during path 
search to narrow down the search space.
For illustration, consider the figure below: Let $du(l_d, l_u, v)$ be
the target def-use pair, its cut points are \{$l_1, l_3, l_d, l_6, l_u$\}.
Here the control point $l_2$ is not a cut point, since the path
$l_1, l_3, l_d, l_4, l_6, l_u$ can be constructed to cover the pair.
For the similar reason, the control points $l_4$ and $l_5$ are not
its cut points. 

\hspace*{3.5cm}\includegraphics[scale=0.6]{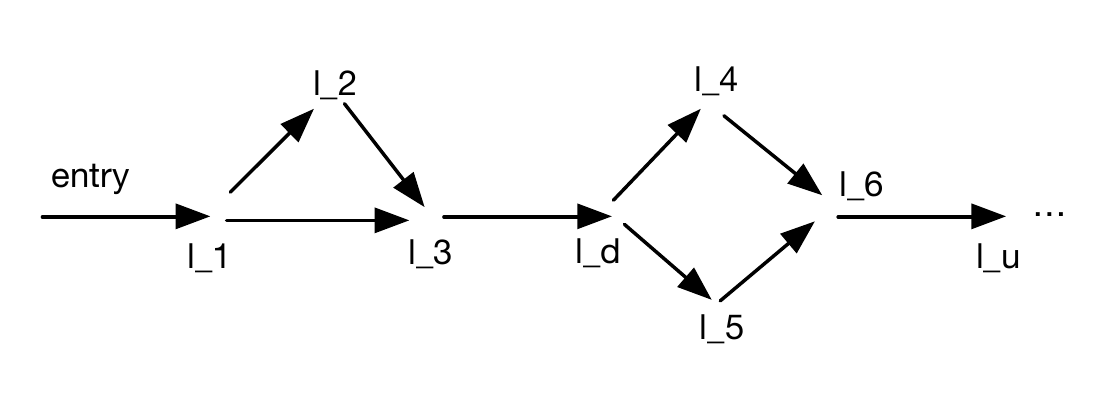}

\SetKwInput{KwInput}{Input}
\SetKwInput{KwInput}{Input}

\SetKw{Let}{let}

\begin{algorithm2e}[t]
\scriptsize
  \setcounter{AlgoLine}{0}
\caption{Cut Point Analysis for Def-use Pairs}\label{algo:cutpoints}
\DontPrintSemicolon
\newcommand\mycommfont[1]{\small\sffamily\textcolor{gray}{#1}}
\SetCommentSty{mycommfont}

\KwIn{$du(l_d,l_u, x)$: a def-use pair}
\KwOut{$C = \{c_1, c_2, \dots, c_n\}$: cut points of $du$}

\BlankLine

$I_d$ $\gets$ getInstruction($l_d$), $I_u$ $\gets$ getInstruction($l_u$), $F$ $\gets$ getFunction($I_d$)\;
\tcp*[h]{Assume $l_d$ and $l_u$ are in the same function}\; 
$C$ $\gets$ getInterpCutPoint($F$,$I_d$) $\cup$ getIntrapCutPoint($I_d$,$I_u$)\;
\tcp*[h]{\emph{getInterpCutPoint} and \emph{getIntrapCutPoint} are the core functions to compute cut points}\; 
\textbf{Procedure} getInterpCutPoint(Function $F$, Instruction $I$)

 \If{$F$ is the main function}{
	getIntrapCP($F.entryInstruction$, $I$)
 }
 
 \Else{
   let $F'$ be the caller function of $F$, \st, $\min \limits_{F' \in callers(F) }$CallgraphDistance($main$, $F'$)\; 
   let ${I}'$ be the callsite instruction of $F$ in $F'$, \st,  $\min \limits_{I' \in callsite(F)}$ InstructionDistance($F'.entryInstruction$,${I}'$)\; 
   getInterCP($F'$, ${I}'$) $\cup$ getIntrapCP($F.entryInstruction$, $I$)\;
 }
 
\textbf{Procedure} getIntrapCutPoint(Instruction $I_t$, Instruction $I_s$)\;
 \lIf{$I_t$ == $I_s$}{\Return $C$} 
 \ElseIf{countPreds($I_s$) == 1}{ \tcp*[h]{$I_s$ has only one pred}\; 
   \If{countSuccs(getPred($I_s$)) == 2}{
      $C \gets C \cup \{I_s\}$\;
   }
   getIntrapCP($I_t$, getPred($I_s$))
 }
 \ElseIf{countPreds($I_s$) $>$ 1}{\tcp*[h]{$I_s$ has more than one pred}\; 
  $I_{idom}$ $\gets$ getIDom($I_s$)\;
	getIntrapCutPoints($I_t$, $I_{idom}$)\;
 }

\tcp*[h]{compute the call graph distance between Function $F$ and $G$}\; 
\textbf{Procedure} CallgraphDistance(Function $F$, Function $G$)\;
$d_{min}$ $\gets$ $\infty$\;
\ForEach{acyclic path $\pi$ on the call graph from $F$ to $G$}{
	$d$ $\gets$ \#function calls of $\pi$ \;
	\Return $d_{min}$ $\gets$ min($d$, $d_{min}$)\;
}

\tcp*[h]{compute the instruction distance between $I$ and $K$ ($I$ and $K$ are in the same function)}\; 
\textbf{Procedure} InstructionDistance(Instruction $I$, Instruction $K$)\;
$d_{min}$ $\gets$ $\infty$\;
\ForEach{acyclic path $\pi$ on the control flow graph from $I$ to $K$}{
	$d$ $\gets$ \#instructions of $\pi$ \;
	\ForEach{function call $f$ on the path $\pi$}{
		\tcp*[h]{distance2Return computes the minimal distance between the entry instruction and one return instruction of function $f$}\; 
		$d$ $\gets$ $d$ + distance2Return($f.entryInstruction$)
	}
	 \Return $d_{min}$ $\gets$ min($d$, $d_{min}$)\;
}
\end{algorithm2e}

Algorithm~\ref{algo:cutpoints} shows the analysis of cut points $C$ \wrt a target pair $du(l_d,l_u, x)$.
To simplify the presentation, Algorithm~\ref{algo:cutpoints} gives the computation algorithm for intra-procedural def-use pairs, \ie, $l_d$ and $l_u$ are in the same function.
For inter-procedural def-use pairs, the computation procedural is similar.
Here, function \emph{getInterpCutPoint} (at Lines 3-9) and \emph{getIntrapCutPoint} (at Lines 10-18) are the two core functions to compute the cut points at the inter- and intra-procedural level, respectively. 
The union of their results is the final set of cut points $C$ (at Line 2).

In particular, the function call \emph{getInterpCutPoint}($F$, $I_d$) (at Line 2) computes the cut points between the \emph{main} function and $l_d$.
If $F$ is the \emph{main} function (\ie, the entry function in C language), $l_d$'s cut points are computed between the entry instruction of $F$ and $I_d$ (at Lines 4-5). Otherwise, it recursively calls \emph{getInterpCutPoint}($F'$, $I'$) (at Lines 7-9), where $F'$ is the caller function of $F$ that satisfies the distance between \emph{main} and $F'$ is minimal on the program call graph, and $I'$ is the callsite instruction of $F$ located in $F'$ that satisfies the distance between the entry intruction of $F'$ and $I'$ is minimal.
In particular, function \emph{CallgraphDistance} (at Lines 19-23) computes the minimal call graph distance between two functions, while function \emph{InstructionDistance} (at Lines 24-30) computes the minimal instruction distance between two instructions in one same function.
We select such $F'$ and $I'$ to facilitate our search strategy (detailed in Section~\ref{sec:se_approach}) can find a valid path to cover the given pair as quickly as possible.

The function call \emph{getIntrapCutPoint}($I_d$, $I_u$) (at Line 2) computes the cut points between $l_d$ and $l_u$.
If the instruction $I_s$ has only one predecessor and this predecessor has two successors, then $I_s$ is a cut point for $I_u$ (at Lines 12-15). If the current instruction $I_s$ has more than one predecessors, we apply dominator analysis in a context-sensitive manner to get its dominator $I_{idom}$, and continue to compute cut points from $I_{idom}$ (at Lines 16-18).

\subsection{SE-based Approach for Data-flow Testing}
\label{sec:se_approach}

This section explains the symbolic execution-based approach for data-flow testing.
Algorithm~\ref{algo:dse_dft} gives the details. This algorithm takes as input a target def-use pair $du$ and its cut points $C$, and either outputs the test case $t$ that satisfies $du$, or $nil$ if it fails to find a path that can cover $du$. 

It first selects one execution state $ES$ from the worklist $W$ which stores all the execution states during symbolic execution. 
It then executes the current program instruction referenced by $ES$, and update $ES$ according to the instruction type (Lines 6-14, \cf  Section~\ref{sec:se_process}). Basically, one instruction can be one of three types: \emph{sequential instruction} (\eg, assignment statements), \emph{forking instruction} (\eg, \texttt{if} statements, denoted as \emph{FORK}), and \emph{exit instruction} (\eg, program exits or runtime errors, denoted as \emph{EXIT}). When it encounters sequential instructions, $ES$ is updated accordingly by function \emph{executeInstruction} (Lines 6-7). Specifically, function \emph{executeInstruction} will internally (1) execute the current instruction, and (2) update $ES$ (including the symbolic state, the reference to next instruction and the corresponding instruction type).
When it encounters \emph{FORK} instructions, one new execution state $ES'$ will be created.
The two states $ES$ and $ES'$ will explore both sides of the fork, respectively, and the corresponding subpaths of $ES$ and $ES'$ will be updated to \emph{$ES$.path+$Fr(T)$} and \emph{$ES$.path+$Fr(F)$}, respectively (Lines 9-14).
Here, $Fr$ denotes the forking point, and $T$ and $F$ represent the \emph{true} and \emph{false} directions, respectively. 
If the target pair $du$ is covered by the pending path $p$ of $ES$, a test input $t$ will be generated (Line 16).
If the variable $x$ of $du$ is redefined on $p$ between the \emph{def} and \emph{use} statement, a redefinition path pruning heuristic will remove those invalid states (Lines 17-18, more details will be explained later).
The algorithm will continue until either the worklist $W$ is empty or the given testing time is exhausted (at Line 19).

\SetKwInput{KwInput}{Input}
\SetKwInput{KwInput}{Input}

\SetKw{Let}{let}

 \begin{algorithm2e}[t]
 \scriptsize
   \setcounter{AlgoLine}{0}
 \caption{SE-based Data-flow Testing}\label{algo:dse_dft}

 \DontPrintSemicolon
 \newcommand\mycommfont[1]{\small\sffamily\textcolor{gray}{#1}}
 \SetCommentSty{mycommfont}

 \KwIn{$du(l_d,l_u, x)$: a given def-use pair}
 \KwIn{$C = \{c_1, c_2, \dots, c_n\}$: the cut points of $du$}
 \KwOut{input $t$ that satisfies $du$ or $\mynil$ if none is found within the given time bound}

 \BlankLine
 \Let $W$ be a worklist of execution states\; 
 \Let $ES_0$ be the initial execution state\;  
 $W$ $\gets$ $W$ $\cup$ \{$ES_0$\}\;
\tcp*[h]{the core process of symbolic execution}\;
 \Repeat{$W$.size()=0 or timeout()}
 {
	 \textbf{ExecutionState} $ES$ $\gets$ selectState($W$)\;
	 \While{$ES$.instructionType!=FORK or EXIT}{
	 	$ES$.executeInstruction()
	 }
	 \lIf{$ES$.instructionType=EXIT}{$W$ $\gets$ $W$ $\setminus$ \{$ES$\}}
	 \If{$ES$.instructionType = FORK}{
	 	\textbf{Instruction} Fr = $ES$.currentInstruction;\;
	 	\textbf{ExecutionState} $ES'$ $\gets$ \textbf{new} executionState($ES$)\;
	 	$ES$.newNode $\gets$ Fr($T$)\;
	 	$ES'$.newNode $\gets$ Fr($F$)\;
	 	$W$ $\gets$ $W$ $\cup$ \{$ES'$\}\;
	 }
	 
	 \textbf{PendingPath} $p$ $\gets$ $ES$.path
	
	 \lIf{$p$ covers $du$}{\Return $t$ $\gets$ getTestCase($ES$)}
	
	 \tcp*[h]{the redefinition path pruning heuristic}\; 
	 \If{variable $x$ (in $du$) is redefined by $p$}{
	 	$W$ $\gets$ $W$ $\setminus$ \{$ES$, $ES'$\}
	 }
  }

 \tcp*[h]{the core algorithm of execution state selection}\;

  \textbf{Procedure} selectState(\textbf{reference} worklist $W$)

  \Let $ES'$ be the next selected execution state\;
	  \tcp*[h]{$j$ is the index of a cut point, $w$ is the state weight}\; 
	   $j$ $\gets$ 0, $w$ $\gets$ $\infty$
	
	  \ForEach{\textbf{ExecutionState} $ES$ $\in$ $W$}{
	
	 \textbf{PendingPath} $pp$ $\gets$ $ES$.path
	 
	 \tcp*[h]{$c_1$, $\dots$, $c_{i}$ are sequentially-covered, while $c_{i+1}$ not yet}
		
	   $i$ $\gets$ index of the cut point $c_i$ on $pp$
	   
	   \textbf{StateWeight} $sw$ $\gets$ distance($es$, $c_{i+1}$)$^{-2}$ + instructionsSinceCovNew($es$)$^{-2}$
	   
	  \If {$i$ $>$ $j$ $\vee$ ($i$ == $j$ $\wedge sw$  $>$ $w$)} 
	 {      $ES'$ $\gets$ $ES$,  $j$ $\gets$ $i$, $w$ $\gets$ $sw$}
	 }
 $W$ $\gets$ $W$ $\setminus$ \{$ES'$\}\; 
 \Return $ES'$
 \end{algorithm2e}

The algorithm core is the state selection procedure, \ie, \emph{selectState} (detailed at Lines 20-30), which integrates several heuristics to improve the overall effectiveness. Fig.~\ref{fig:hybrid_search} conceptually shows the benefits of their combination (the red path is a valid path that covers the pair), which can efficiently steer exploration towards the target pair, and reduce as many unnecessary path explorations as possible: (1) the cut point guided search guides the state exploration towards the target pair more quickly;
(2) the backtrack strategy counts the number of executed instructions to prevent the search from being trapped in tight loops, and switches to alternative search directions; and (3) the redefinition path pruning technique effectively prunes redundant search space. 
In detail, we use Formula~\ref{eq:hybrid_search_formula} to assign the weights to all states, and achieve the heuristics (1) and (2).

\begin{align}\label{eq:hybrid_search_formula}
state\_weight(es)=(c_{max}, \frac{1}{d^2}+\frac{1}{i^2})
\end{align}
where, $ES$ is an execution state, $c_{max}$ is the deepest covered cut point, $d$ is the instruction distance toward the next uncovered cut point, and $i$ is the number of executed instructions since the last new instruction have been covered.
Below, we explain the details of each heuristic.

\vspace*{2pt}
\noindent{\textbf{\emph{Cut point guided search}}}. \quad
The cut-point guided search strategy (at Lines 23-28) aims to search for the ES whose pending path has covered the deepest cut point, and tries to reach the next goal, \ie, the next uncovered cut point, as quickly as possible.
For an ES, its pending path is a subpath that starts from the program entry and reaches up to the program location of it.
If this path has sequentially covered the cut point $c_1, c_2, \ldots, c_i$ but not $c_{i+1}$, $c_{i}$ is the deepest covered cut point, and $c_{i+1}$ is the next goal to reach.
The strategy always prefers to select the ES that has covered the deepest cut point (at Lines 26-28, indicated by the condition $i > j$).
The intuition is that the deeper cut point an ES can reach, the closer the ES toward the pair is.

\begin{figure}[t]
	\begin{center}
		\includegraphics[width=0.6\textwidth]{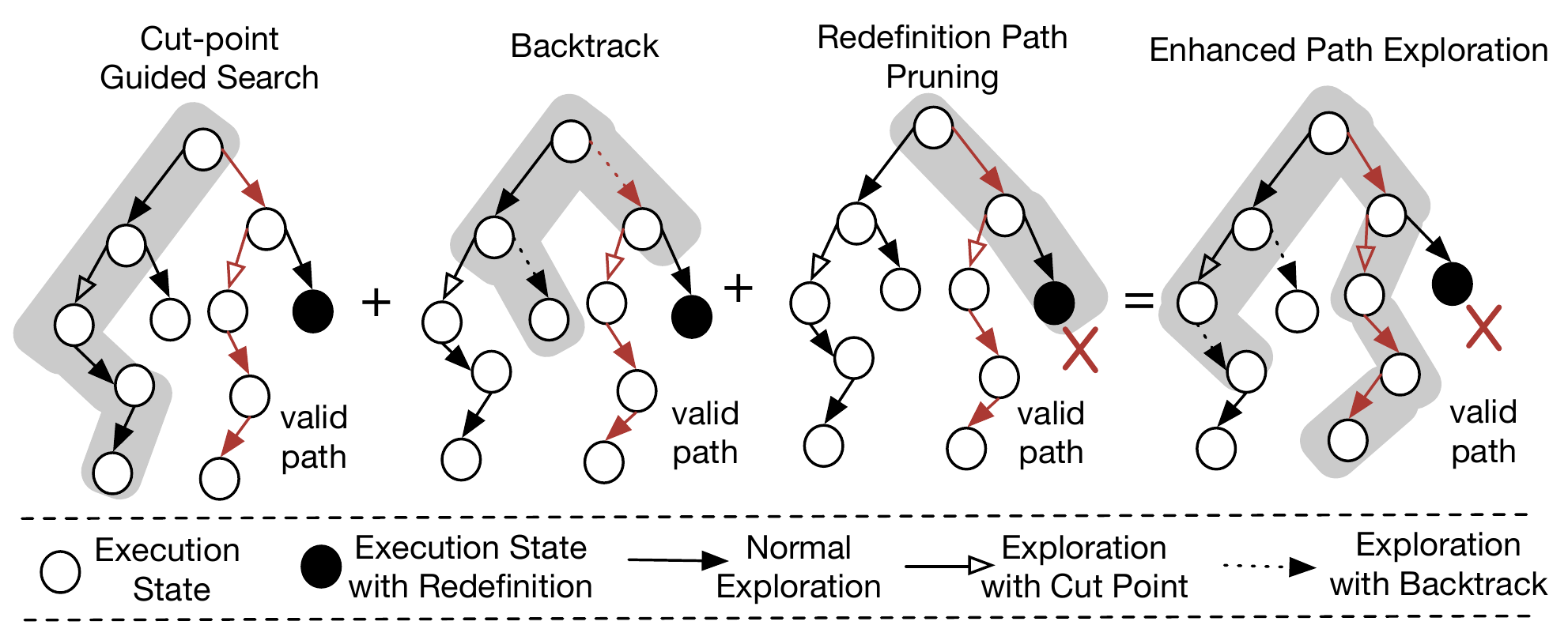}
	\end{center}
	\caption{Enhanced path exploration in symbolic execution: combine cut-point guided search, backtrack strategy and redefinition path pruning. Each subfigure denotes the execution tree generated by symbolic execution.}
	\label{fig:hybrid_search}
\end{figure}

When more than one ES covers the deepest cut point 
(indicated by the condition $i$==$j$ at Line 27), 
the ES that has the shortest distance toward next goal will be preferred (at Lines 26-28).
The intuition is that the closer the distance is, the more quickly the ES can reach the goal.
We use \emph{dist($es$, $c_{i+1}$)} to present the distance 
between the location of $es$ and the next uncovered cut point $c_{i+1}$.
The distance is approximated as the number of instructions along the shortest control-flow path between the program locations of $es$ and $c_{i+1}$.

\vspace*{2pt}
\noindent{\textbf{\emph{Backtrack strategy}}}. \quad
To avoid the execution falling into the tight loops, we assign an ES with lower priority if the ES is not likely to cover new instructions.
In particular, for each ES, the function \emph{instrsSinceCovNew}, corresponding to $i$ in Formula (\ref{eq:hybrid_search_formula}), counts the number of executed instructions since the last new instruction is covered (at Line 26).
The ES, which has a larger value of \emph{instrsSinceCovNew}, is assumed that it has lower possibility to cover new instructions.
Intuitively, this heuristic prefers the ES which is able to cover more new instructions, if a ES does not cover new instructions for a long time, the strategy will backtrack to another ES via lowering the weight of the current ES.

\vspace*{2pt}
\noindent{\textbf{\emph{Redefinition Path Pruning}}}. \quad
A redefinition path pruning technique checks whether the selected ES has redefined the variable $x$ in $du$.
If the ES is invalid (\ie, its pending path has redefined $x$), it will be discarded and \emph{selectState} will choose another one (at Lines 17-18).
The reason is that, according to the definition of DFT (\cf Definition~\ref{defn:dft}), 
it is impossible to find def-clear paths by executing those invalid ESs. 

Further, by utilizing the path-sensitive information from SE, we can detect variable redefinitions, especially caused by variable aliases, more precisely.
Variable aliases appear when two or more variable names refer to the same memory location.
So we designed a lightweight variable redefinition detection algorithm in our framework.
Our approach operates upon a simplified three-address form of the original code\footnote{
We use CIL as the C parser to transform the source code into an equivalent simplified form 
using the --dosimplify option, where one statement contains at most one operator.
}, 
so we mainly focus on the following statement forms where variable aliases and variable redefinitions may appear:

\begin{itemize}
  \item Alias inducing statements: (1) $p$:=$q$ ($*p$ is an alias to $*q$), (2) $p$:=\&$x$ ($*p$ is an alias to $x$)
  \item Variable definition statements: (3) $*p$:=$y$ ($*p$ is defined by $y$),  (4) $v$:=$y$ ($v$ is defined by $y$)
\end{itemize}
Here, $p$ and $q$ are pointer variables, $x$ and $y$ non-pointer variables, and ``:=" the assignment operator.

Initially, a set $A$ is maintained, which denotes the variable alias set \wrt the variable $x$ of $du$. At the beginning, it only contains $x$ itself.
During path exploration, if the executed statement is (1) or (2), and $*q$ or $x$
$\in A$, $*p$ will be added into $A$ since $*p$ becomes an alias of $x$.
If the executed statement is (1), and $*q \not\in A$ but $x \in A$, $*p$ will be excluded
from $A$ since it becomes an alias of another variable instead of $x$.
If the executed statement is (3) or (4), and $*p \in A$ or $x \in A$, 
the variable is redefined by another variable $y$.

\subsection{CEGAR-based Approach for Data-flow Testing}
\label{sec:mc_approach}

\begin{figure}[t]
\begin{center}
\includegraphics[width=0.5\textwidth]{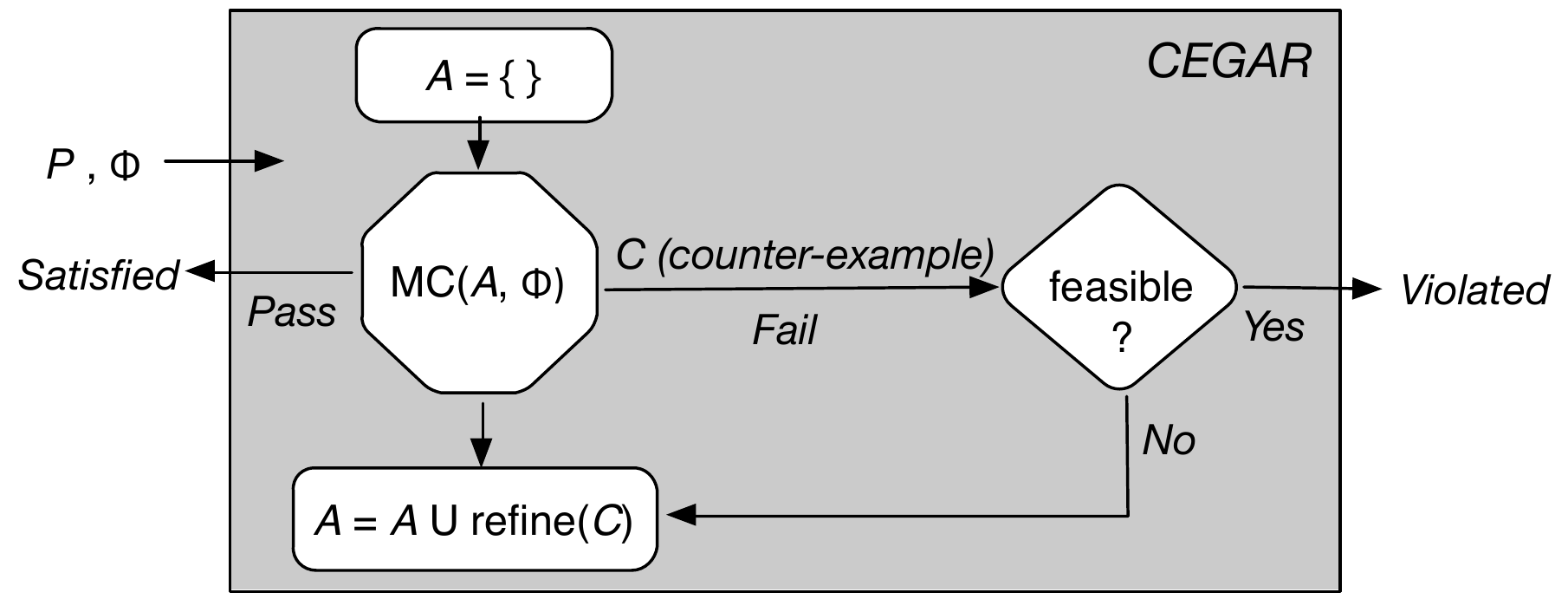}
\end{center}
    \caption{Paradigm of CEGAR-based Model Checking}
\label{fig:cegar_workflow}
\end{figure}

Counter-example guided abstract refinement (CEGAR) is a well-known software model checking technique that statically proves program correctness \wrt properties (or specifications) of interest~\cite{jhala2009software}.
Fig.~\ref{fig:cegar_workflow} shows the basic paradigm of CEGAR, which typically follows an \emph{abstract-check-refine} paradigm.
Given the program $P$ (\ie, the actual implementation) and a safety property $\phi$ of interest, 
CEGAR first \emph{abstracts} $P$ into a model $A$ (typically represented
as a finite automaton), and then \emph{checks} the property $\phi$ against $A$.
If the abstract model $A$ is error-free, then so is the original program $P$.
If it finds a path on the model $A$ that violates the property $\phi$, it will check the feasibility of this path: 
is it a genuine path that can correspond to a concrete path
in the original program $P$, or due to the result of the current coarse abstraction?
If the path is feasible, CEGAR returns a counter-example path $C$ to demonstrate the violation of $\phi$.
Otherwise, CEGAR will utilize this path $C$ to \emph{refine} $A$ by adding new predicates, 
and continue the checking until it either finds a genuine path that violates $\phi$ or proves that $\phi$ 
is always satisfied in $P$. 
Or since this model checking problem itself is undecidable, CEGAR does not terminate and cannot conclude the correctness of $\phi$.

To exploit the power of CEGAR, our approach reduces the problem of data flow testing to the problem of model checking.
The CEGAR-based approach can operate in two phases~\cite{Beyer04} to generate tests, 
\ie, \emph{model checking} and \emph{tests from counter-examples}.
(1) It first uses model checking to check whether the specified program location $l$ is reachable such that 
the predicate of interest $q$ (\ie, the safety property) can be satisfied at that point.
(2) If $l$ is reachable, CEGAR will return a counter-example path $p$ that establishes $q$ at $l$,
and generate a test case from the corresponding path constraint of $p$.
Otherwise, if $l$ is not reachable, CEGAR will conclude no test inputs can reach $l$.

The key idea is to encode the test obligation of a target def-use pair into the program under test.
We instrument the original program $P$ to $P'$, and reduce the test generation for $P$ to path reachability checking on $P'$.
In particular, we follow three steps:
(1) We introduce a variable \emph{cover\_flag} into $P$, which denotes the cover status of the target pair, and initialize it as \emph{false}.
(2) The variable \emph{cover\_flag} is set as \emph{true} immediately after the \emph{def} statement, 
and set as \emph{false} immediately after all redefinition statements.
(3) Before the \emph{use} statement, we set the target predicate as \emph{cover\_flag==true}.
As a result, if the use statement is reachable when the target predicate holds, we 
can obtain a counter-example (\ie, a test case) and conclude the pair is feasible.
Otherwise, if unreachable, we can safely conclude that the pair is infeasible (or since
the problem itself is undecidable, the algorithm does not terminate, and gives \emph{unknown}).

\vspace*{3pt}
\noindent\emph{\textbf{Generability of the CEGAR-based approach}}. \quad
This reduction approach is flexible to implement on any CEGAR-based model checkers.
It is also applicable for other software model checking techniques, \eg, Bounded Model Checking (BMC).
Given a program, BMC unrolls the control flow graph for a fixed number of $k$ steps, and checks whether the property $p$ at a specified program location $l$ is violated or not. Different from the modern (dynamic) symbolic execution techniques, BMC executes on pure symbolic inputs without using any concrete input values, and usually aims to systematically checking reachability within given bounds. Different from CEGAR, BMC searches on all program computations without abstraction and typically backtracks the search when a given (loop or search depth) bound is reached. le
Although BMC in general cannot prove infeasibility as certain (unless any inductive reasoning mechanism~\cite{MorseRCN014,GadelhaIC17} is used), in Section~\ref{sec:evaluation} we will show the BMC-based approach can still serve as a heuristic-criterion to identify hard-to-cover (probably infeasible) pairs and particularly effective for specific types of programs. In fact, the infeasible pairs concluded by BMC can be regarded as valid modulo the given checking bounds.

\section{Proof}
\label{sec:proof}

Symbolic execution is an explicit path-based analysis technique, where a path exploration algorithm
should be specified to determine the search priority of different paths.
In this paper, we designed a cut-point guided search algorithm to enhance SE for
DFT.
Thus, it is desirable to prove the correctness of this algorithm, and show its effectiveness.

\begin{theorem}[Correctness of cut point guided search]
Any path that covers the target def-use pair must pass through its cut points in succession; any path
that passes through the cut points of a def-use pair in succession can cover the pair.
\end{theorem}

\begin{proof}
Let $du(l_d, l_u, x)$ be the target def-use pair, and its cut points are \{$l_1, \ldots, l_i, \ldots, l_j, \ldots, l_n$\}  
($n$ is the number of cut points of $du$). 
According to the definition of DFT (\cf Definition~\ref{defn:dft}), if there exists a path $p$ that covers $du$, 
the cut point $l_d$ and $l_u$ is passed through by $p$.

Assume there exists a cut point $l_i$ that is not passed through by $p$. 
According to the definition of cut point (\cf Definition~\ref{def:cutpoint}), the cut point $l_i$ dominates $l_u$
($l_u$ is the last cut point of $du$), as a result, any path that reaches $l_u$ must pass through $l_i$.
As a result, $p$ passes through $l_i$, however, it is contradictory with the assumption.
So there does not exist such a cut point $l_i$, \ie, all cut points are passed through in succession by $p$.
On the other hand, if there exists a path $p$ that passes through the cut points of $du$ in succession, 
$l_d$ and $l_u$ are covered (assume no redefinitions appear between $l_d$ and $l_u$), and thus, $du$ is
covered by $p$.
\end{proof}

\begin{theorem}[Soundness of the combined approach]
Our approach combines SE and SMC (CEGAR in particular). For a feasible def-use pair, it generates a test case to
satisfy it; for an infeasible pair, it gives an infeasibility proof. In theory, our approach is sound that it will not generate
any false positives.
\end{theorem}

\begin{proof}
According to Algorithm~\ref{algo:dse_dft}, the SE-based approach explicitly explores program 
paths by enumerating execution states in $W$.
When SE finds a test case that can satisfy a target def-use pair, 
this pair is concluded feasible.
The CEGAR-based approach directly instruments the test requirement of a pair into the program under test, and 
reduces data flow testing into path reachability checking.
When CEGAR finds a concrete counter-example path in the original program
that can reach the \emph{use} statement and establish the validity of the target predicate at
the same time,
the pair is concluded as feasible with a corresponding test case;
when CEGAR proves the \emph{use} statement is unreachable, a proof is produced.
During this process, our approach will not bring any false positives --- the feasibility of 
a def-use pair is justified by a concrete test case or a proof of infeasibility.
But note our approach may draw \emph{unknown} conclusions due to the undecidability 
of the problem itself, and under this circumstance, both SE and CEGAR may not terminate.
\end{proof}

\section{Framework Design and Implementation}
\label{sec:tool}
We realized our hybrid data-flow testing framework for \texttt{C} programs.
In our original work~\cite{SuFPHS15}, we implemented the SE-based approach on our own concolic testing based tool CAUT~\cite{SuPFHYJZ14,CAUT}, while in this article we built the enhanced SE-based approach on KLEE~\cite{KLEE}, a robust execution-generated testing based symbolic execution engine, to fully exhibit its feasibility. 
As for the SMC-based approach, we instantiated it with
two different types of software model checking techniques, \ie, CEGAR and BMC. 
In all, our framework combines the SE-based and SMC-based approaches together to achieve efficient DFT.

In the static analysis phase, we identify def-use pairs, cut points, and related static program information
(\eg, variable definitions and aliases)
by using CIL~\cite{CIL} (C Intermediate Language), which is an infrastructure for C 
program analysis and transformation.
We first build the control-flow graph (CFG) for each function in the program under test,
and then construct the inter-procedural CFG (ICFG) for the whole program.
We perform standard iterative data-flow analysis techniques~\cite{Harrold94,Pande94}, \ie,
reaching definition analysis, to compute def-use pairs.
For each variable use, we compute which definitions on the same variable can reach the use
through a def-clear path on the control-flow graph.
A def-use pair is created as a test objective for each use with its corresponding definition.
We treat each formal parameter variable as defined at the beginning of its function and each argument parameter variable as used at its function call site (\eg, library calls).
For global variables, we treat them as initially defined at the beginning the entry function (\eg, the \emph{main} function), and defined/used at any function where they are defined/used.

In the current implementation, we focus on the def-use pairs with local variables (intra-procedural pairs) and global variables (inter-procedural pairs).
Following prior work on data-flow testing~\cite{VivantiMGF13}, we currently do not consider the def-use pairs induced by pointer aliases. Thus, we may miss some def-use pairs, but we believe that this is an independent issue (not the focus of this work) and does not affect the effectiveness of our testing approach.
More sophisticated data-flow analysis techniques (\eg, dynamic data-flow analysis~\cite{DenaroPV14}) or 
tools (\eg, Frama-C~\cite{KirchnerKPSY15}, SVF~\cite{SuiX16}) can be used to mitigate this problem.
 
Specifically, to improve the efficiency of state selection algorithm (\cf Algorithm~\ref{algo:dse_dft}) in KLEE, we use the priority queue to sort execution states according to their weights. The algorithmic complexity is $\mathcal{O}(n\log{}n)$ ($n$ is the number of execution states), which is much faster than using a list or array.
The software model checkers are used as black-box to enforce data-flow testing. The benefit of this design choice is that we can flexibly integrate any model checker without any modification or adaption.
CIL transforms the program under test into a simplified code version, and encodes the test requirements of def-use pairs into the program under test.
Both SE-based and SMC-based tools takes as input the same CIL-simplified code.
Function stubs are used to simulate C library functions such as string, memory and file operations to improve the ability of symbolic reasoning.
To compute the data-flow coverage during testing, we implement the classic \emph{last definition} technique~\cite{Horgan92} in KLEE. We maintain a table of def-use pairs, and insert probes at each basic block to monitor the program execution. The runtime routine records each variable that has been defined and the block where it was defined. 
When a block that uses this defined variable is executed, the last definition of this variable is located, we check whether the pair is covered. Our implementations are publicly available at~\cite{dft-artifact}.
\section{Evaluation}
\label{sec:evaluation}
This section aims to evaluate whether our combined testing approach can achieve efficient data-flow testing. In particular, we intend to investigate (1) whether the core SE-based approach can quickly cover def-use pairs; (2) whether the SMC-based reduction approach is feasible and practical; and (3) whether the combined approach can be more effective for data-flow testing.

\subsection{Research Questions}
\label{sec:rqs_and_studies}

\begin{itemize}[leftmargin=*]
	\item \emph{\textbf{RQ1}}: In the data-flow testing \wrt all def-use coverage, what is the performance difference between different existing search strategies (\eg, DFS, RSS, RSS-MD2U, SDGS) and CPGS (our cut point guided path search strategy) in terms of testing time and number of covered pairs for the SE-based approach? 
	
	\item \emph{\textbf{RQ2}}: How is the practicability of the CEGAR-based reduction approach as well as the BMC-based approach in terms of testing time and number of identified feasible and infeasible pairs?
	\item \emph{\textbf{RQ3}}: How efficient is the combined approach, which complements the SE-based approach with the SMC-based approach, in terms of testing time and coverage level, compared with the SE-based approach or the SMC-based approach alone? 
	
\end{itemize}

\subsection{Evaluation Setup}
\label{sec:setup}

\vspace*{3pt}
\noindent{\emph{\textbf{Testing Environment}}}. \quad
All evaluations were run on a 64bit Ubuntu 14.04 physical machine with 24 processors (2.60GHz Intel Xeon(R) E5-2670 CPU) and 94GB RAM.

\vspace*{3pt}
\noindent{\emph{\textbf{Framework Implementations}}}. \quad
The SE-based approach of our hybrid testing framework was implemented on KLEE (v1.1.0), and the SMC-based approach was implemented on two different software model checking techniques, CEGAR and BMC.
In particular, we chose three different software model checkers\footnote{We use the latest versions of these model checkers at the time of writing.}, \ie, BLAST~\cite{Beyer07} (CEGAR-based, v2.7.3), CPAchecker~\cite{BeyerK11} (CEGAR-based, v.1.7) and CBMC~\cite{ckl2004} (BMC-based, v5.7). 
We chose different model checkers, since we intend to gain more overall understandings of the practicality of this reduction approach.
Note that the CEGAR-based approach can give definite answers of the feasibility, while the BMC-based approach is used as a heuristic-criterion to identify hard-to-cover (probably infeasible) pairs. 

\vspace*{2pt}
\noindent{\emph{\textbf{Program Subjects}}}. \quad 
Despite data-flow testing has been continuously investigated in the past four decades, the standard benchmarks for evaluating data-flow testing techniques are still missing. 
To this end, we took substantial efforts and dedicatedly constructed a repository of benchmark subjects by following these steps.
\emph{First}, we collected the subjects from prior work on data-flow testing. We conducted a thorough investigation on all prior work (99 papers~\cite{dft_repo} in total) related to data-flow testing, and searched for the adopted subjects. After excluding the subjects whose source codes are not available or not written in C language, we got 26 unique subjects from 19  papers~\cite{eler2014covering,GirgisS14,Marre03,HutchinsFGO94,Khamis11,Singla11,Singla11a,NayakM10,GhidukHG07,
Girgis05,MerloA99,WongM95,MathurW94,Frankl93,RappsW85,Weyuker84,Rapps82,SuFPHS15}. We then manually inspected these programs and excluded 19 subjects which are too simple, we finally got 7 subjects. These 7 subjects include mathematical computations and classic algorithms.
\emph{Second}, we included 7 Siemens subjects from SIR~\cite{SIR}, which are widely used in the experiments of program analysis and software testing~\cite{HutchinsFGO94,HassanA13,WangLGSZY17}. These subjects involve numeral computations, string manipulations and complex data structures (\eg, \texttt{pointer}s, \texttt{struct}s, and \texttt{list}s).
\emph{Third}, we further enriched the repository with the subjects from the SV-COMP benchmarks~\cite{SV-COMP}, which are originally used for the competition on software verification.
The SV-COMP benchmarks are categorized in different groups by their features (\eg, concurrency, bit vectors, floats) for evaluating software model checkers.
In order to reduce potential evaluation biases in our scenario, we carefully inspected all the benchmarks and finally decided to select subjects from the ``Integers and Control Flow'' category based on these considerations: (1) the subjects in this category are real-world (medium-sized or large-sized) OS device drivers (\cf~\cite{TACAS12-SVCOMP}, Section 4), while many subjects in other categories are hand-crafted, small-sized programs; (2) the subjects in this category have complicated function call chains or control-flow structures, which are more appropriate for our evaluation; (3) the subjects do not contain specific features that may not be supported by KLEE (\eg, concurrency, floating point numbers). We finally selected 16 subjects in total from the \emph{ntdrivers} and \emph{ssh} groups therein (we excluded other subjects with similar control-flow structures by diffing the code).
The selected subjects have rather complex control-flows. For example, the average cyclomatic complexity of functions in the \texttt{ssh} group exceeds 88.5\footnote{Cyclomatic complexity is a software metric that indicates the complexity of a program. The standard software development guidelines recommend the cyclomatic complexity of a module should not exceeded 10.} (computed by \emph{Cyclo}~\cite{cyclo}).
\emph{Fourth}, we also included three core program modules from the industrial projects from our research partners.
The first one is an engine management system (\emph{osek\_control}) running on an automobile operating system conforming to the OSEK/VDX standard.
The second one is a satellite gesture control program (\emph{space\_control}).
The third one is a control program (\emph{subway\_control}) from a subway signal.
All these three industrial programs were used in our previous research work~\cite{PengHSG13,SuPFHYJZ14,MiaoPYSB0CX16,Zhang2018}, and have complicated data-flow interactions.
\emph{Finally}, to ensure each tool (where our framework is built upon) can correctly reason these subjects, we carefully read the documentation of each tool to understand their limitations, manually checked each program and added necessary function stubs (\eg, to simulate such C library functions as string, memory, and file operations) but without affecting their original program logic and structures. This is important to reduce validation threats, and also provides a more fair comparison basis.
Totally, we got 33 subjects with different characteristics, including mathematical computation, classic algorithms, utility programs, device drivers and industrial control programs.
These subjects allow us to evaluate diverse data-flow scenarios. 
Table~\ref{table:table_subjects} shows the detailed statistics of these subjects, which includes the executable lines of code (computed by \emph{cloc}~\cite{cloc}), the number of def-use pairs (including intra- and inter-procedural pairs), and the brief functional description. 

\begin{table}[t]
\scriptsize
\newcommand{\tabincell}[2]{\begin{tabular}{@{}#1@{}}#2\end{tabular}}
\renewcommand{\arraystretch}{1}
\caption{Subjects of the constructed data-flow testing benchmark repository}
\label{table:table_subjects}
\centering
\begin{tabular}{|c||c|c|c|}
\hline
Subject &\#ELOC &\tabincell{c}{\#DU\\} &Description\\
\hline \hline
\emph{factorization} &43 &47 &compute factorization \\
\emph{power} &11  &11  & compute the power $x^y$ \\
\emph{find} &66 &99 &permute an array's elements \\
\emph{triangle} &32 &46 &classify an triangle type\\
\emph{strmat} &67 &32 &string pattern matching\\
\emph{strmat2} &88 &38 &string pattern matching \\
\emph{textfmt} &142 &73 &text string formatting \\ \hline

\emph{tcas} &195 &86  &collision avoidance system\\
\emph{replace} &567 &387 &pattern matching and substitution\\
\emph{totinfo} &374 &279 &compute statistics given input data \\
\emph{printtokens} &498 &240 &lexical analyzer\\
\emph{printtokens2} &417 &192 & lexical analyzer\\
\emph{schedule} &322 &118 & process priority scheduler \\
\emph{schedule2} &314 &107 & process priority scheduler \\ \hline

\emph{kbfiltr} &557 &176 &\texttt{ntdrivers} group \\
\emph{kbfiltr2} &954 &362 &\texttt{ntdrivers} group \\
\emph{diskperf} &1,052 &443 &\texttt{ntdrivers} group \\
\emph{floppy} &1,091 &331 &\texttt{ntdrivers} group \\
\emph{floppy2} &1,511 &606 &\texttt{ntdrivers} group \\
\emph{cdaudio} &2,101 &773 &\texttt{ntdrivers} group \\

\emph{s3\_clnt} &540 &1,677 &\texttt{ssh} group \\
\emph{s3\_clnt\_termination} &555 &1,595 &\texttt{ssh} group \\
\emph{s3\_srvr\_1a} &198 &574 &\texttt{ssh} group \\
\emph{s3\_srvr\_1b} &127 &139 &\texttt{ssh} group \\
\emph{s3\_srvr\_2} &608 &2,130 &\texttt{ssh} group \\
\emph{s3\_srvr\_7} &624 &2,260 &\texttt{ssh} group \\
\emph{s3\_srvr\_8} &631 &2,322 &\texttt{ssh} group \\
\emph{s3\_srvr\_10} &628 &2,200 &\texttt{ssh} group \\
\emph{s3\_srvr\_12} &696 &3,125 &\texttt{ssh} group \\
\emph{s3\_srvr\_13} &642 &2,325 &\texttt{ssh} group \\ \hline

\emph{osek\_control} &4,589 &927 & one module of engine management system \\
\emph{space\_control} &5,782 &1,739 & one module of satellite gesture control software\\
\emph{subway\_control} &5,612 &2,895 & one module of subway signal control software\\
\hline
\end{tabular}
\end{table}

\vspace*{3pt}
\noindent{\emph{\textbf{Search Strategies for Comparison}}}. \quad To our knowledge, there exists no specific guided search strategies on KLEE to compare with our strategy. Thus, we compare our cut-point guided search strategy with several existing search strategies.
In particular, we chose two generic search strategies (\ie, depth-first and random search), one popular (statement) coverage-optimized search strategy. In addition, we implemented one search strategy for directed testing on KLEE, which is proposed by prior work~\cite{ZamfirC10,MaKFH11,MarinescuC13}. We detail them as follows. 
\begin{itemize}[leftmargin=*]
  
   \item \emph{Depth First-Search (DFS)}: always select the latest execution state from all states to explore, and has little overhead in state selection.
   \item \emph{Random State Search (RSS)}: randomly select an execution state from all states to explore, and able to explore the program space more uniformly and less likely to be trapped in tight loops than other strategies like DFS.
   \item \emph{Coverage-Optimized Search (COS)}: compute the weights of the states by some heuristics, \eg, the minimal distance to uncovered instructions (\emph{md2u}) and whether the state recently covered new code (\emph{covnew}), and randomly select states \wrt these weights. These heuristics are usually interleaved with other search strategies in a round-robin fashion to improve their overall effectiveness.
   For example, \emph{RSS-COS:md2u} (\emph{RSS-MD2U} for short) is a popular strategy used by KLEE,
   which interleaves \emph{RSS} with \emph{md2u}.
   \item \emph{Shortest Distance Guided Search (SDGS)}: always select the execution state that has the shortest (instruction) distance toward a target instruction in order to cover the target as quickly as possible. This strategy has been widely applied in single target testing~\cite{ZamfirC10,MaKFH11,MarinescuC13}. In the context of data-flow testing, we implemented this strategy in KLEE by setting the \emph{def} as the first goal and then the \emph{use} as the second goal after the \emph{def} is covered.
\end{itemize}

\subsection{Case Studies}
\label{sec:rqs_and_studies}

We conducted three case studies to answer the research questions.
Note that in this paper we focus on the classic data-flow testing~\cite{GhidukHG07,VivantiMGF13}, \ie, targeting one def-use pair at one run.
In \emph{\textbf{Study 1}}, we answer \emph{\textbf{RQ1}} by comparing the performance of different search strategies that were implemented on KLEE.
In detail, we use two metrics: (1) \emph{number of covered pairs}, \ie, how many def-use pairs can be covered; and (2) \emph{testing time}, \ie, how long does it take to cover the pair(s) of interest. The \emph{testing time} is measured by the median value and the semi-interquartile range (SIQR)\footnote{SIQR = (Q3-Q1)/2, which measures the variability of testing time, where Q1 is the lower quartile, and Q3 is the upper quartile.} of the times consumed on those covered (\emph{feasible}) pairs\footnote{In theory, the symbolic execution-based approach cannot identify infeasible pairs unless it enumerates all possible paths, which however is impossible in practice. Therefore, we only consider the testing times of covered (feasible) pairs for performance evaluation.}.

In the evaluation, the maximum allowed search time on each pair is set as 5 minutes. Under this setting, we observed all search strategies can thoroughly test each subject (\ie, reach their highest coverage rates). To mitigate the algorithm randomness, we repeat the testing process 30 times for each program/strategy and aggregate their average values as the final results for all measurements.

In \emph{\textbf{Study 2}}, we answer \emph{\textbf{RQ2}} by evaluating the practicability of the SMC-based reduction approach on two different model checking techniques, CEGAR and BMC. 
Specifically, we implemented the reduction approach on three different model checkers, BLAST, CPAchecker and CBMC.
We use the following default command options and configurations according to their user manuals
and the suggestions from the tool developers, respectively:

\begin{footnotesize}
\begin{description}
\item[\textbf{BLAST}:]
\begin{verbatim}
     ocamltune blast -enable-recursion -cref -lattice -noprofile -nosserr -quiet
\end{verbatim}
\item[\textbf{CPAchecker}:]
\begin{verbatim}
     cpachecker -config config/predicateAnalysis.properties -skipRecursion
\end{verbatim}
\item[\textbf{CBMC}:]
\begin{verbatim}
     cbmc --slice-formula --unwind nr1 --depth nr2
\end{verbatim}
\end{description}
\end{footnotesize}
We have not tried to particularly tune the optimal configurations of these tools for different subjects under test, since we aim to investigate the practicability of our reduction approach in general.  
Specifically, BLAST and CPAchecker are configured based on predicate abstraction.
For BLAST, we use an internal script \texttt{ocamltune} to improve memory utilization for large programs; for CPAchecker, we use its default predicate abstraction configuration \textsf{\small{predicateAnalysis.properties}}. 
We use the option \textsf{\small{-enable-recursion}} of BLAST and \textsf{\small{-skipRecursion}} of CPAchecker to set recursion functions as \emph{skip}.
Due to CBMC is a bounded model checker, it may answer \emph{infeasible} for actual \emph{feasible} pairs if the given checking bound is too small.
Thus, we set the appropriate values for the \textsf{\small{--unwind}} and \textsf{\small{--depth}} options, respectively, for the number of times loops to be unwound and the number of program steps to be processed. Specially, we determine the parameter values of \textsf{\small{--unwind}} and \textsf{\small{--depth}} options by a binary search algorithm to ensure that CBMC can identify as many pairs as possible within the given time bound. This avoids wasting testing budget on unnecessary path explorations, and also achieves a more fair evaluation basis. Therefore, each subject may be given different parameter values (the concrete parameter values of all subjects are available at~\cite{dft-artifact}).

Specifically, we use two metrics: (1) \emph{number of feasible, infeasible, and unknown pairs}; and (2) \emph{testing (checking) time of feasible and infeasible pairs} (denoted in medians).
The maximum testing time on each def-use pair is constrained as 5 minutes (\ie, 300 seconds, the same setting in RQ1). For each def-use pair, we also run 30 times to mitigate algorithm randomness.

In \emph{\textbf{Study 3}}, we answer \emph{\textbf{RQ3}} by combining the SE-based and SMC-based approaches. We interleave these two approaches as follows: 
the SE-based approach (configured with the cut point-guided path search strategy and the same settings in \emph{\textbf{RQ1}}) 
is first used to cover as many pairs as possible; then, for the remaining uncovered pairs, the SMC-based approach (configured with the same settings in \emph{\textbf{RQ2}}) is used to identify infeasible pairs (may also cover some feasible pairs). 
We continue the above iteration of the combined approach until the maximum allowed time bound (5 minutes for each pair) is used up. 
Specifically, we increase the time bound by 3 times at each iteration, \ie, 10s, 30s, 90s and 300s.

Specifically, we use two metrics: (1) \emph{coverage rate}; and (2) \emph{total testing time}, \ie, the total time required to enforce data-flow testing on all def-use pairs of one subject.
The coverage rate $C$ is computed by Formula~\ref{eq:coverage_rate}, where \emph{nTestObj} is the total number of pairs, and \emph{nFeasible} and \emph{nInfeasible} are the number of identified feasible and infeasible ones, respectively. 

\begin{footnotesize}
	\begin{align}
	C &=\frac{nFeasible}{nTestObj-nInfeasible}\times100\% \label{eq:coverage_rate}
	\end{align}
\end{footnotesize}

In all case studies, the testing time was measured in CPU time via the \emph{time} command in Linux. In particular, the testing time did not include IO operations for logging the testing results. We tested 31,634 ELOC with 28,354 pairs in total.
It took us nearly one and half months to run the experiments and analyze the results.

\subsection{Study 1}
\label{sec:study1}

\begin{table*}[t]
\newcommand{\tabincell}[2]{\begin{tabular}{@{}#1@{}}#2\end{tabular}}
\centering
\caption{Performance statistics of different search strategies for data-flow testing (the testing time is measured in seconds).}
\label{table:RQ1_search_strategy_testing}
\resizebox{\linewidth}{!}{
\def\arraystretch{1.2} 
\begin{tabular}{| c | c c | cc | cc | cc| cc|}
\toprule
 \multicolumn{1}{|c|}{\textbf{Subject}} 
 &\multicolumn{2}{c|}{\textbf{DFS}} 
 &\multicolumn{2}{c|}{\textbf{RSS}} 
 &\multicolumn{2}{c|}{\textbf{RSS-MD2U}}
 &\multicolumn{2}{c|}{\textbf{SDGS}}
 &\multicolumn{2}{c|}{\textbf{CPGS}}
 \\
 
&N & M (SIQR) 
&N & M (SIQR) 
&N & M (SIQR) 
&N & M (SIQR) 
&N & M (SIQR)  

\\ \hline
\tabincell{c}{\emph{factorization}}
& \tabincell{c}{22}
& \tabincell{c}{0.07 (0.01)}
& \tabincell{c}{22}
& \tabincell{c}{0.07 (0.01)}
& \tabincell{c}{22}
& \tabincell{c}{0.08 (0.02)}
& \tabincell{c}{22}
& \tabincell{c}{\underline{0.05} (0.01)}
& \tabincell{c}{22}
& \tabincell{c}{0.06 (0.01)}

\\ \hline
\tabincell{c}{\emph{power}}
& \tabincell{c}{6}
& \tabincell{c}{0.14 (0.00)}
& \tabincell{c}{9}
& \tabincell{c}{0.12 (0.01)}
& \tabincell{c}{9}
& \tabincell{c}{0.05 (0.01)}
& \tabincell{c}{5}
& \tabincell{c}{\underline{0.04} (0.00)}
& \tabincell{c}{9}
& \tabincell{c}{\underline{0.04} (0.00)}

\\ \hline
\tabincell{c}{\emph{find}}
& \tabincell{c}{77}
& \tabincell{c}{0.89(0.64)}
& \tabincell{c}{49}
& \tabincell{c}{\underline{0.19} (0.54)}
& \tabincell{c}{52}
& \tabincell{c}{0.26 (0.31)}
& \tabincell{c}{51}
& \tabincell{c}{0.63 (3.35)}
& \tabincell{c}{56}
& \tabincell{c}{0.22 (0.12)}

\\ \hline
\tabincell{c}{\emph{triangle}}
& \tabincell{c}{22}
& \tabincell{c}{0.24 (0.06)}
& \tabincell{c}{22}
& \tabincell{c}{0.24 (0.03)}
& \tabincell{c}{22}
& \tabincell{c}{0.26 (0.05)}
& \tabincell{c}{22}
& \tabincell{c}{0.25 (0.09)}
& \tabincell{c}{22}
& \tabincell{c}{\underline{0.13} (0.01)}

\\ \hline
\tabincell{c}{\emph{strmat}}
& \tabincell{c}{26}
& \tabincell{c}{2.84 (1.41)}
& \tabincell{c}{30}
& \tabincell{c}{\underline{0.10} (0.02)}
& \tabincell{c}{30}
& \tabincell{c}{0.13 (0.03)}
& \tabincell{c}{30}
& \tabincell{c}{0.12 (0.16)}
& \tabincell{c}{30}
& \tabincell{c}{\underline{0.10} (0.02)}

\\ \hline
\tabincell{c}{\emph{strmat2}}
& \tabincell{c}{28}
& \tabincell{c}{2.85 (1.40)}
& \tabincell{c}{32}
& \tabincell{c}{\underline{0.09} (0.01)}
& \tabincell{c}{32}
& \tabincell{c}{0.11 (0.02)}
& \tabincell{c}{32}
& \tabincell{c}{0.11 (0.03)}
& \tabincell{c}{32}
& \tabincell{c}{\underline{0.09} (0.02)}

\\ \hline
\tabincell{c}{\emph{textfmt}}
& \tabincell{c}{37}
& \tabincell{c}{0.16 (0.08)}
& \tabincell{c}{33}
& \tabincell{c}{\underline{0.05} (0.01)}
& \tabincell{c}{33}
& \tabincell{c}{0.11 (0.04)}
& \tabincell{c}{34}
& \tabincell{c}{0.06 (0.01)}
& \tabincell{c}{34}
& \tabincell{c}{0.06 (0.01)}

\\ \hline
\tabincell{c}{\emph{tcas}}
& \tabincell{c}{55}
& \tabincell{c}{\underline{0.13} (0.03)}
& \tabincell{c}{55}
& \tabincell{c}{0.21 (0.07)}
& \tabincell{c}{55}
& \tabincell{c}{0.67 (0.43)}
& \tabincell{c}{55}
& \tabincell{c}{0.16 (0.06)}
& \tabincell{c}{55}
& \tabincell{c}{0.14 (0.06)}

\\ \hline
\tabincell{c}{\emph{replace}}
& \tabincell{c}{69}
& \tabincell{c}{\underline{0.77} (0.14)}
& \tabincell{c}{308}
& \tabincell{c}{1.96 (15.23)}
& \tabincell{c}{312}
& \tabincell{c}{30.31 (21.97)}
& \tabincell{c}{295}
& \tabincell{c}{4.67 (5.46)}
& \tabincell{c}{309}
& \tabincell{c}{1.15 (3.58)}

\\ \hline
\tabincell{c}{\emph{totinfo}}
& \tabincell{c}{13}
& \tabincell{c}{0.52 (0.07)}
& \tabincell{c}{24}
& \tabincell{c}{\underline{0.42} (0.13)}
& \tabincell{c}{24}
& \tabincell{c}{0.64 (0.08)}
& \tabincell{c}{24}
& \tabincell{c}{\underline{0.42} (0.06)}
& \tabincell{c}{26}
& \tabincell{c}{0.52 (0.05)}

\\ \hline
\tabincell{c}{\emph{printtokens}}
& \tabincell{c}{48}
& \tabincell{c}{\underline{0.96} (0.62)}
& \tabincell{c}{115}
& \tabincell{c}{34.69 (24.16)}
& \tabincell{c}{106}
& \tabincell{c}{33.68 (22.59)}
& \tabincell{c}{107}
& \tabincell{c}{16.40 (25.53)}
& \tabincell{c}{115}
& \tabincell{c}{12.23 (20.21)}

\\ \hline
\tabincell{c}{\emph{printtokens2}}
& \tabincell{c}{124}
& \tabincell{c}{\underline{0.47} (0.32)}
& \tabincell{c}{148}
& \tabincell{c}{0.80 (3.72)}
& \tabincell{c}{149}
& \tabincell{c}{20.67 (18.42)}
& \tabincell{c}{149}
& \tabincell{c}{0.83 (3.48)}
& \tabincell{c}{154}
& \tabincell{c}{0.51 (1.43)}

\\ \hline
\tabincell{c}{\emph{schedule}}
& \tabincell{c}{15}
& \tabincell{c}{\underline{0.16} (0.03)}
& \tabincell{c}{83}
& \tabincell{c}{0.23 (3.98)}
& \tabincell{c}{86}
& \tabincell{c}{0.76 (5.05)}
& \tabincell{c}{77}
& \tabincell{c}{0.22 (1.67)}
& \tabincell{c}{86}
& \tabincell{c}{0.22 (1.84)}

\\ \hline
\tabincell{c}{\emph{schedule2}}
& \tabincell{c}{14}
& \tabincell{c}{\underline{0.15} (0.02)}
& \tabincell{c}{78}
& \tabincell{c}{0.20 (0.12)}
& \tabincell{c}{78}
& \tabincell{c}{0.48 (1.11)}
& \tabincell{c}{77}
& \tabincell{c}{0.21 (0.10)}
& \tabincell{c}{77}
& \tabincell{c}{0.21 (0.08)}

\\ \hline
\tabincell{c}{\emph{cdaudio}}
& \tabincell{c}{562}
& \tabincell{c}{3.13 (0.41)}
& \tabincell{c}{562}
& \tabincell{c}{3.27 (0.48)}
& \tabincell{c}{562}
& \tabincell{c}{15.54 (7.11)}
& \tabincell{c}{562}
& \tabincell{c}{3.77 (2.52)}
& \tabincell{c}{562}
& \tabincell{c}{\underline{3.08} (0.51)}

\\ \hline
\tabincell{c}{\emph{diskperf}}
& \tabincell{c}{285}
& \tabincell{c}{\underline{0.89} (0.19)}
& \tabincell{c}{302}
& \tabincell{c}{0.97 (0.21)}
& \tabincell{c}{302}
& \tabincell{c}{1.97 (4.80)}
& \tabincell{c}{299}
& \tabincell{c}{0.95 (0.23)}
& \tabincell{c}{302}
& \tabincell{c}{0.92 (0.18)}

\\ \hline
\tabincell{c}{\emph{floppy}}
& \tabincell{c}{249}
& \tabincell{c}{\underline{0.62} (0.11)}
& \tabincell{c}{249}
& \tabincell{c}{0.67 (0.11)}
& \tabincell{c}{249}
& \tabincell{c}{2.11 (1.76)}
& \tabincell{c}{249}
& \tabincell{c}{0.72 (0.14)}
& \tabincell{c}{249}
& \tabincell{c}{0.66 (0.13)}

\\ \hline
\tabincell{c}{\emph{floppy2}}
& \tabincell{c}{510}
& \tabincell{c}{2.22 (0.37)}
& \tabincell{c}{510}
& \tabincell{c}{2.14 (0.42)}
& \tabincell{c}{510}
& \tabincell{c}{6.44 (3.44)}
& \tabincell{c}{510}
& \tabincell{c}{3.62 (1.66)}
& \tabincell{c}{510}
& \tabincell{c}{\underline{2.03} (0.39)}

\\ \hline

\tabincell{c}{\emph{kbfiltr}}
& \tabincell{c}{116}
& \tabincell{c}{\underline{0.26} (0.05)}
& \tabincell{c}{116}
& \tabincell{c}{0.28 (0.05)}
& \tabincell{c}{116}
& \tabincell{c}{0.49 (0.41)}
& \tabincell{c}{116}
& \tabincell{c}{0.31 (0.04)}
& \tabincell{c}{116}
& \tabincell{c}{0.27 (0.05)}

\\ \hline
\tabincell{c}{\emph{kbfiltr2}}
& \tabincell{c}{266}
& \tabincell{c}{0.97 (0.15)}
& \tabincell{c}{266}
& \tabincell{c}{0.94 (0.18)}
& \tabincell{c}{266}
& \tabincell{c}{4.18 (3.51)}
& \tabincell{c}{266}
& \tabincell{c}{2.11 (1.08)}
& \tabincell{c}{266}
& \tabincell{c}{\underline{0.90} (0.20)}

\\ \hline
\tabincell{c}{\emph{s3\_srvr\_1a}}
& \tabincell{c}{113}
& \tabincell{c}{0.72 (0.17)}
& \tabincell{c}{171}
& \tabincell{c}{0.75 (0.19)}
& \tabincell{c}{171}
& \tabincell{c}{0.76 (0.19)}
& \tabincell{c}{165}
& \tabincell{c}{0.65 (0.17)}
& \tabincell{c}{171}
& \tabincell{c}{\underline{0.58} (0.18)}

\\ \hline
\tabincell{c}{\emph{s3\_srvr\_1b}}
& \tabincell{c}{30}
& \tabincell{c}{0.08 (0.02)}
& \tabincell{c}{43}
& \tabincell{c}{0.08 (0.02)}
& \tabincell{c}{43}
& \tabincell{c}{\underline{0.07} (0.02)}
& \tabincell{c}{43}
& \tabincell{c}{\underline{0.07} (0.02)}
& \tabincell{c}{45}
& \tabincell{c}{0.08 (0.02)}

\\ \hline
\tabincell{c}{\emph{s3\_clnt}}
& \tabincell{c}{647}
& \tabincell{c}{\underline{9.64} (2.01)}
& \tabincell{c}{648}
& \tabincell{c}{11.93 (2.64)}
& \tabincell{c}{647}
& \tabincell{c}{22.91 (8.19)}
& \tabincell{c}{633}
& \tabincell{c}{12.45 (2.45)}
& \tabincell{c}{648}
& \tabincell{c}{10.32 (1.88)}

\\ \hline
\tabincell{c}{\emph{s3\_clnt\_termination}}
& \tabincell{c}{333}
& \tabincell{c}{9.20 (1.64)}
& \tabincell{c}{332}
& \tabincell{c}{8.81 (1.87)}
& \tabincell{c}{332}
& \tabincell{c}{12.14 (1.29)}
& \tabincell{c}{332}
& \tabincell{c}{9.56 (1.74)}
& \tabincell{c}{414}
& \tabincell{c}{\underline{6.54} (1.03)}

\\ \hline
\tabincell{c}{\emph{s3\_srvr\_2}}
& \tabincell{c}{414}
& \tabincell{c}{\underline{14.35} (2.67)}
& \tabincell{c}{695}
& \tabincell{c}{24.23 (17.42)}
& \tabincell{c}{695}
& \tabincell{c}{31.86 (15.44)}
& \tabincell{c}{681}
& \tabincell{c}{19.93 (8.00)}
& \tabincell{c}{695}
& \tabincell{c}{16.45 (3.76)}

\\ \hline
\tabincell{c}{\emph{s3\_srvr\_7}}
& \tabincell{c}{420}
& \tabincell{c}{\underline{16.29} (3.44)}
& \tabincell{c}{710}
& \tabincell{c}{27.82 (20.06)}
& \tabincell{c}{710}
& \tabincell{c}{34.99 (17.11)}
& \tabincell{c}{686}
& \tabincell{c}{26.93 (12.46)}
& \tabincell{c}{815}
& \tabincell{c}{19.47 (5.41)}

\\ \hline
\tabincell{c}{\emph{s3\_srvr\_8}}
& \tabincell{c}{416}
& \tabincell{c}{\underline{16.77} (3.10)}
& \tabincell{c}{704}
& \tabincell{c}{23.61 (14.03)}
& \tabincell{c}{698}
& \tabincell{c}{36.15 (16.39)}
& \tabincell{c}{690}
& \tabincell{c}{23.45 (7.21)}
& \tabincell{c}{798}
& \tabincell{c}{17.04 (4.23)}

\\ \hline
\tabincell{c}{\emph{s3\_srvr\_10}}
& \tabincell{c}{431}
& \tabincell{c}{\underline{15.26} (2.40)}
& \tabincell{c}{683}
& \tabincell{c}{21.34 (5.19)}
& \tabincell{c}{683}
& \tabincell{c}{30.21 (7.03)}
& \tabincell{c}{664}
& \tabincell{c}{20.73 (5.85)}
& \tabincell{c}{683}
& \tabincell{c}{18.37 (3.90)}

\\ \hline
\tabincell{c}{\emph{s3\_srvr\_12}}
& \tabincell{c}{433}
& \tabincell{c}{\underline{25.76} (3.84)}
& \tabincell{c}{395}
& \tabincell{c}{39.51 (21.68)}
& \tabincell{c}{539}
& \tabincell{c}{64.25 (38.99)}
& \tabincell{c}{486}
& \tabincell{c}{39.50 (18.04)}
& \tabincell{c}{724}
& \tabincell{c}{25.88 (10.08)}

\\ \hline
\tabincell{c}{\emph{s3\_srvr\_13}}
& \tabincell{c}{437}
& \tabincell{c}{\underline{15.69} (2.07)}
& \tabincell{c}{489}
& \tabincell{c}{25.25 (18.07)}
& \tabincell{c}{558}
& \tabincell{c}{33.78 (21.20)}
& \tabincell{c}{572}
& \tabincell{c}{23.98 (11.49)}
& \tabincell{c}{744}
& \tabincell{c}{15.77 (6.12)}

\\ \hline
\tabincell{c}{\emph{osek\_control}}
& \tabincell{c}{398}
& \tabincell{c}{{7.69} (2.47)}
& \tabincell{c}{426}
& \tabincell{c}{15.77 (14.68)}
& \tabincell{c}{549}
& \tabincell{c}{23.32 (17.39)}
& \tabincell{c}{538}
& \tabincell{c}{14.17 (6.17)}
& \tabincell{c}{639}
& \tabincell{c}{\underline{6.15} (3.23)}

\\ \hline
\tabincell{c}{\emph{space\_control}}
& \tabincell{c}{582}
& \tabincell{c}{{15.90} (7.76)}
& \tabincell{c}{812}
& \tabincell{c}{33.49 (20.61)}
& \tabincell{c}{990}
& \tabincell{c}{48.77 (23.08)}
& \tabincell{c}{961}
& \tabincell{c}{28.86 (15.78)}
& \tabincell{c}{1,178}
& \tabincell{c}{\underline{6.32} (7.09)}

\\ \hline
\tabincell{c}{\emph{subway\_control}}
& \tabincell{c}{827}
& \tabincell{c}{{13.44} (7.69)}
& \tabincell{c}{967}
& \tabincell{c}{42.76 (28.65)}
& \tabincell{c}{1,290}
& \tabincell{c}{68.61 (31.73)}
& \tabincell{c}{1,244}
& \tabincell{c}{38.11 (21.46)}
& \tabincell{c}{1,654}
& \tabincell{c}{\underline{10.72} (6.72)}

\\ \hline \hline
\tabincell{c}{\textbf{Total}}
& \tabincell{c}{8,025}
& -
& \tabincell{c}{10,088}
& -

& \tabincell{c}{10,920}
& -
& \tabincell{c}{10,677}
& -
& \tabincell{c}{12,240}
& -

\\

\bottomrule
\end{tabular}}
\end{table*}

Table~\ref{table:RQ1_search_strategy_testing} shows the detailed performance statistics of different search strategies. The column \emph{Subject} represents the subject under test, DFS, RSS, RSS-MD2U, SDGS, CPGS, respectively, represent the search strategies. For each subject/strategy, it shows the number of covered def-use pairs (denoted by \emph{N}), the median value of testing times (denoted by \emph{M}) and the semi-interquartile range of testing times (denoted by \emph{SIQR}) on all covered pairs. 
In particular, for each subject, we underscore the strategy with lowest median value.
The last row gives the total number of covered pairs.
From Table~\ref{table:RQ1_search_strategy_testing}, we can observe
(1) Given enough testing time for all strategies (\ie, 5 minutes for each pair), CPGS covers 4215, 2152, 1320 and 1563 more pairs, respectively, than DFS, RSS, RSS-MD2U and SDGS.
It means CPGS achieves 40\%, 21.3\%, 12.1\%, 14.6\% higher data-flow coverage than these strategies, respectively.
(2) By comparing the median values of CPGS with those of other strategies, CPGS achieves more efficient data-flow testing in 14/33, 23/33, 32/33, 26/33 subjects than DFS, RSS, RSS-MD2U and SDGS, respectively. Note that the median value of DFS is low because it only covers many easily reachable pairs, which also explains why it achieves lowest coverage.

\begin{figure}[t]
	\begin{center}
		\includegraphics[width=0.48\textwidth]{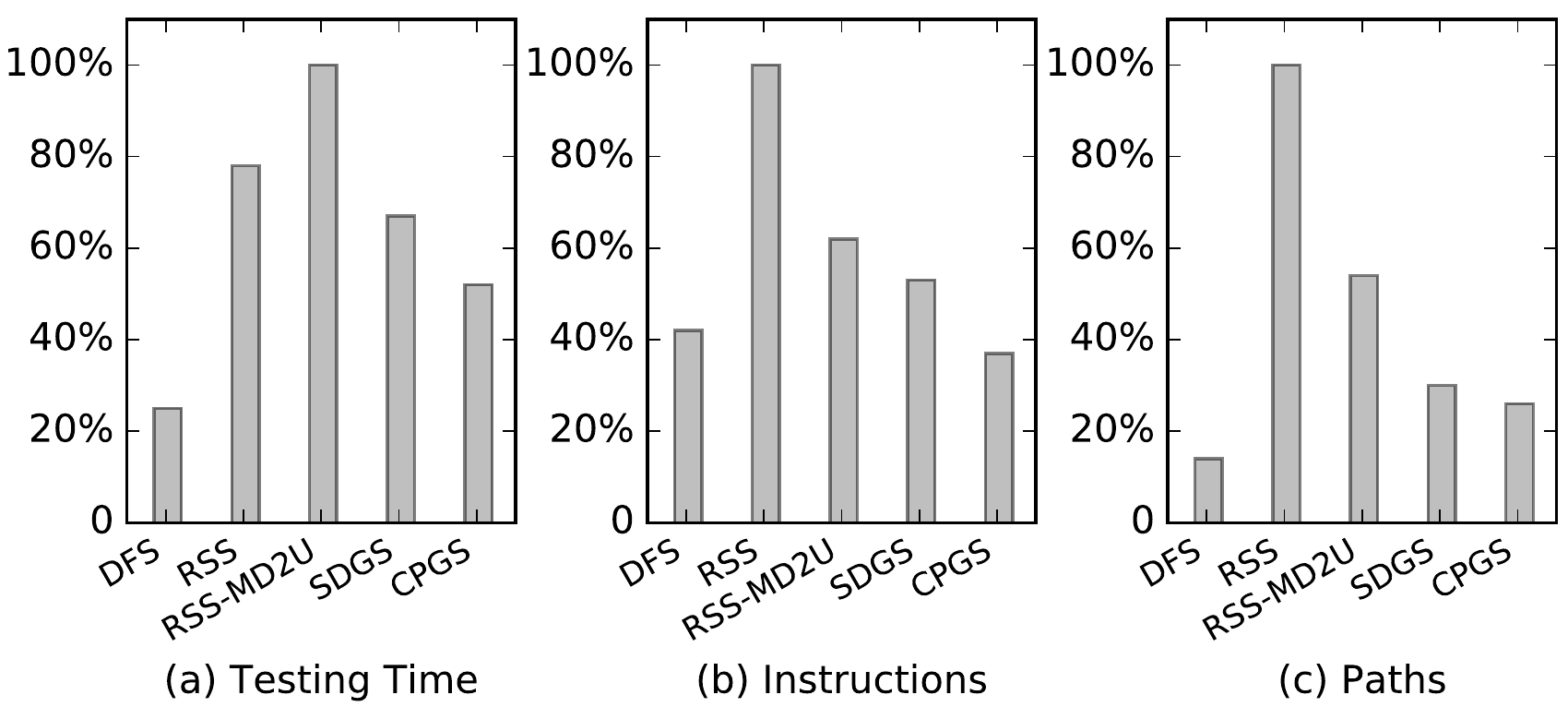}
	\end{center}
	\caption{Performance of each search strategy in terms of total testing time, number of executed program instructions, and number of explored program paths (normalized in percentage) on all 33 subjects.}
	\label{fig:se_strategies_comparision}
\end{figure}

Fig.~\ref{fig:se_strategies_comparision} shows the performance of these search strategies on all 33 subjects in terms of total testing time, the number of executed program instructions, and the number of explored program paths (due to the data difference, we normalized them in percentage). 
Note these three metrics are all computed on the covered pairs.
Apart from DFS (since it achieves rather low data-flow coverage), we can see CPGS outperforms all the other testing strategies.
In detail, CPGS reduces testing time by 15$\sim$48\%, the number of executed instructions by 16$\sim$63\%, and the number of explored paths by 28$\sim$74\%. The reason is that CPGS narrows down the search space by following the cut points and prunes unnecessary paths.

\begin{table*}[t]
	\newcommand{\tabincell}[2]{\begin{tabular}{@{}#1@{}}#2\end{tabular}}
	\centering
	\caption{Comparison of CPGS \wrt SDGS and CPGS for the testing time on each pair, where * denotes \emph{p < 0.05} and ** denotes \emph{p < 0.01}. $\hat{A}_{12}$ denotes Vargha-Delaney effect size.}
	\label{table:RQ1_stats_testing}
	\resizebox{0.6\linewidth}{!}{
		\begin{tabular}{| c | cc  | cc | }
			\toprule
			\multicolumn{1}{|c|}{\textbf{Subject}} 
			&\multicolumn{2}{c|}{\textbf{CPGS \vs SDGS}} 
			&\multicolumn{2}{c|}{\textbf{CPGS \vs RSS-MD2U}} 
			
			\\
			&\#Common Pairs & $\hat{A}_{12}$
			&\#Common Pairs & $\hat{A}_{12}$
			
			\\ \hline
			\tabincell{c}{\emph{replace}}
			& 292
			& \tabincell{c}{0.64*}
			& 309
			& \tabincell{c}{0.77**}
			
			\\ \hline
			\tabincell{c}{\emph{printtokens}}
			& 106
			& \tabincell{c}{0.61*}
			& 106
			& \tabincell{c}{0.65*}
			
			\\ \hline
			\tabincell{c}{\emph{printtokens2}}
			& 149
			& \tabincell{c}{0.65**}
			& 149
			& \tabincell{c}{0.86**}
			
			\\ \hline
			\tabincell{c}{\emph{cdaudio}}
			& 562
			& \tabincell{c}{0.57*}
			& 562
			& \tabincell{c}{0.98**}
			
			\\ \hline
			\tabincell{c}{\emph{floppy2}}
			& 510
			& \tabincell{c}{0.62*}
			& 510
			& \tabincell{c}{0.95**}
			
			\\ \hline
			\tabincell{c}{\emph{kbfilter2}}
			& 264
			& \tabincell{c}{0.62*}
			& 264
			& \tabincell{c}{0.93**}
			
			\\ \hline
			\tabincell{c}{\emph{s3\_clnt}}
			& 633
			& \tabincell{c}{0.65**}
			& 647
			& \tabincell{c}{0.92**}
			
			\\ \hline
			\tabincell{c}{\emph{s3\_clnt\_termination}}
			& 332
			& \tabincell{c}{0.62*}
			& 332
			& \tabincell{c}{0.91**}
			
			\\ \hline
			\tabincell{c}{\emph{s3\_srvr\_2}}
			& 680
			& \tabincell{c}{0.71**}
			& 691
			& \tabincell{c}{0.84**}
			
			\\ \hline
			\tabincell{c}{\emph{s3\_srvr\_7}}
			& 683
			& \tabincell{c}{0.68**}
			& 706
			& \tabincell{c}{0.82**}
			
			\\ \hline
			\tabincell{c}{\emph{s3\_srvr\_8}}
			& 686
			& \tabincell{c}{0.65**}
			& 690
			& \tabincell{c}{0.84**}
			
			\\ \hline
			\tabincell{c}{\emph{s3\_srvr\_10}}
			& 662
			& \tabincell{c}{0.64**}
			& 676
			& \tabincell{c}{0.84**}
			
			\\ \hline
			\tabincell{c}{\emph{s3\_srvr\_12}}
			& 485
			& \tabincell{c}{0.75**}
			& 539
			& \tabincell{c}{0.85**}
			
			\\ \hline
			\tabincell{c}{\emph{s3\_srvr\_13}}
			& 568
			& \tabincell{c}{0.72**}
			& 553
			& \tabincell{c}{0.83**}
			
			\\ \hline
			\tabincell{c}{\emph{osek\_control}}
			& 533
			& \tabincell{c}{0.73**}
			& 541
			& \tabincell{c}{0.88**}
			
			\\ \hline
			\tabincell{c}{\emph{space\_control}}
			& 957
			& \tabincell{c}{0.75**}
			& 979
			& \tabincell{c}{0.92**}
			
			\\ \hline
			\tabincell{c}{\emph{subway\_control}}
			& 1,237
			& \tabincell{c}{0.81**}
			& 1,279
			& \tabincell{c}{0.93**}
			
			\\
			\bottomrule
\end{tabular}}
\end{table*}

Table~\ref{table:RQ1_stats_testing} further evaluates the statistical significance for the testing time on each pair among CPGS, SDGS and RSS-MD2U (the three best strategies indicated by Table~\ref{table:RQ1_search_strategy_testing} and Fig.~\ref{fig:se_strategies_comparision}) on the 17 complicated subjects. Specifically, we pair-wisely compare CPGS with SDGS, and CPGS with RSS-MD2U (denoted by the columns \emph{CPGS vs. SDGS}, and \emph{CPGS vs. RSS-MD2U}), respectively, on their commonly covered pairs. We can also note CPGS covers almost all the pairs covered by SDGS or RSS-MD2U.
Following the guidelines of statistical tests~\cite{ArcuriB14}, we use Wilcoxon T test (\ie, Wilcoxon signed-rank test, a non-parametric statistical hypothesis test for paired samples), and Vargha-Delaney $\hat{A}_{12}$ effect size, to compare these strategies.
The test result in Table~\ref{table:RQ1_stats_testing} suggestions that CPGS significantly reduces the testing time on each pair (\emph{p-value<0.01} or \emph{p-value<0.05}) \wrt SDGS and RSS-MD2U. 
Moreover, CPGS outperforms SDGS with large effect size on 6/17 subjects and medium effective size on 6/17 subjects; CPGS outperforms RSS-MD2U with large effective size on 16/17 subjects.
Based on these analysis, we can conclude that CPGS performs the best in the data-flow testing.

Additionally, we also find some interesting phenomenons worth elaborating: (1) RSS has lower testing time than RSS-MD2U due to its lower state scheduling overhead (see Fig.~\ref{fig:se_strategies_comparision}).
However, RSS-MD2U covers 832 more pairs than RSS (see Table~\ref{table:RQ1_search_strategy_testing}). 
Compared with RSS, RSS-MD2U reduces the number of executed instructions and explored paths by 38\% and 46\%, respectively (see Fig.~\ref{fig:se_strategies_comparision}).
It indicates, by combining coverage-optimized heuristics, RSS-MD2U can indeed improve the effectiveness of data-flow testing, but at the same time it incurs more state scheduling overhead.
(2) SDGS is comparable to CPGS in the number of explored paths, but much less effective than CPGS (1320 fewer pairs). The reason is that SDGS only uses the instruction distance as the guidance. When SDGS always chooses to explore the directions with shorter distance, it can easily be trapped in tight program loops. On the other hand, the backtrack strategy integrated in CPGS can help alleviate this issue.
(3) DFS costs least testing time but also covers least number of pairs among all the strategies. In rare cases, due to the characteristics of program, DFS may cover more pairs than the other strategies (see \emph{find} in Table~\ref{table:RQ1_search_strategy_testing}).

\vspace{2pt}
\noindent\fbox{
	\parbox{0.95\linewidth}{
		\textbf{Answer to RQ1:} \textit{In summary, our cut-point guided search (CPGS) strategy performs the best for data-flow testing. It improves 12$\sim$40\% data-flow coverage, and at the same time reduces the total testing time by 15$\sim$48\% and the number of explored paths by 28$\sim$74\%. Therefore, the SE-based approach, enhanced with the cut point guided search strategy, is efficient for data-flow testing.} 
	}
}

\subsection{Study 2}
\label{sec:study2}

\begin{table*}[!htbp]
\scriptsize
\newcommand{\tabincell}[2]{\begin{tabular}{@{}#1@{}}#2\end{tabular}}
\centering
\caption{Performance statistics of the SMC-based reduction approach $CEGAR_{BLAST}$, $CEGAR_{CPAchecker}$ and $BMC_{CBMC}$ for data-flow testing (the testing time is measured in seconds), where * denotes the numbers in the corresponding columns are only valid modulo the given checking bound for $BMC_{CBMC}$.}
\label{tab:model_checker_testing}
\resizebox{\linewidth}{!}{
\begin{tabular}{| c | c c c c c | c c c c c || c c c c c |}
\toprule
 \multicolumn{1}{|c|}{\textbf{Subject}} 
 &\multicolumn{5}{c|}{\textbf{$CEGAR_{BLAST}$}} 
 &\multicolumn{5}{c||}{\textbf{$CEGAR_{CPAchecker}$}} 
 &\multicolumn{5}{c|}{\textbf{$BMC_{CBMC}$}} 
 \\

&F &I &U & $M_F$ & $M_I$
&F &I &U & $M_F$ & $M_I$
&F &I* &U* & $M_F$ & $M_I$*

\\ \hline

\tabincell{c}{\emph{factorization}}
& \tabincell{c}{35}
& \tabincell{c}{4}
& \tabincell{c}{8}
& \tabincell{c}{0.04}
& \tabincell{c}{0.20}
& \tabincell{c}{26}
& \tabincell{c}{4}
& \tabincell{c}{17}
& \tabincell{c}{3.26}
& \tabincell{c}{3.04}
& \tabincell{c}{41}
& \tabincell{c}{{6}}
& \tabincell{c}{0}
& \tabincell{c}{0.34}
& \tabincell{c}{0.28}
\\ \hline

\tabincell{c}{\emph{power}}
& \tabincell{c}{9}
& \tabincell{c}{2}
& \tabincell{c}{0}
& \tabincell{c}{0.03}
& \tabincell{c}{0.49}
& \tabincell{c}{9}
& \tabincell{c}{2}
& \tabincell{c}{0}
& \tabincell{c}{3.10}
& \tabincell{c}{2.97}
& \tabincell{c}{9}
& \tabincell{c}{2}
& \tabincell{c}{0}
& \tabincell{c}{0.13}
& \tabincell{c}{0.12}
\\ \hline

\tabincell{c}{\emph{find}}
& \tabincell{c}{85}
& \tabincell{c}{12}
& \tabincell{c}{2}
& \tabincell{c}{6.44}
& \tabincell{c}{3.22}
& \tabincell{c}{74}
& \tabincell{c}{14}
& \tabincell{c}{11}
& \tabincell{c}{4.37}
& \tabincell{c}{3.60}
& \tabincell{c}{77}
& \tabincell{c}{{22}}
& \tabincell{c}{0}
& \tabincell{c}{0.29}
& \tabincell{c}{0.29}
\\ \hline

\tabincell{c}{\emph{triangle}}
& \tabincell{c}{22}
& \tabincell{c}{24}
& \tabincell{c}{0}
& \tabincell{c}{0.04}
& \tabincell{c}{0.69}
& \tabincell{c}{22}
& \tabincell{c}{24}
& \tabincell{c}{0}
& \tabincell{c}{3.09}
& \tabincell{c}{2.83}
& \tabincell{c}{22}
& \tabincell{c}{24}
& \tabincell{c}{0}
& \tabincell{c}{0.11}
& \tabincell{c}{0.11}
\\ \hline

\tabincell{c}{\emph{strmat}}
& \tabincell{c}{30}
& \tabincell{c}{2}
& \tabincell{c}{0}
& \tabincell{c}{1.81}
& \tabincell{c}{1.39}
& \tabincell{c}{30}
& \tabincell{c}{2}
& \tabincell{c}{0}
& \tabincell{c}{4.67}
& \tabincell{c}{2.98}
& \tabincell{c}{30}
& \tabincell{c}{2}
& \tabincell{c}{0}
& \tabincell{c}{0.15}
& \tabincell{c}{0.15}
\\ \hline

\tabincell{c}{\emph{strmat2}}
& \tabincell{c}{32}
& \tabincell{c}{6}
& \tabincell{c}{0}
& \tabincell{c}{5.08}
& \tabincell{c}{1.46}
& \tabincell{c}{32}
& \tabincell{c}{6}
& \tabincell{c}{0}
& \tabincell{c}{4.91}
& \tabincell{c}{3.79}
& \tabincell{c}{32}
& \tabincell{c}{6}
& \tabincell{c}{0}
& \tabincell{c}{0.15}
& \tabincell{c}{0.15}
\\ \hline

\tabincell{c}{\emph{textfmt}}
& \tabincell{c}{47}
& \tabincell{c}{18}
& \tabincell{c}{8}
& \tabincell{c}{10.08}
& \tabincell{c}{13.90}
& \tabincell{c}{53}
& \tabincell{c}{20}
& \tabincell{c}{0}
& \tabincell{c}{12.69}
& \tabincell{c}{5.50}
& \tabincell{c}{53}
& \tabincell{c}{20}
& \tabincell{c}{0}
& \tabincell{c}{3.84}
& \tabincell{c}{3.95}
\\ \hline

\tabincell{c}{\emph{tcas}}
& \tabincell{c}{55}
& \tabincell{c}{31}
& \tabincell{c}{0}
& \tabincell{c}{1.35}
& \tabincell{c}{1.31}
& \tabincell{c}{55}
& \tabincell{c}{31}
& \tabincell{c}{0}
& \tabincell{c}{4.08}
& \tabincell{c}{3.43}
& \tabincell{c}{55}
& \tabincell{c}{31}
& \tabincell{c}{0}
& \tabincell{c}{0.13}
& \tabincell{c}{0.13}
\\ \hline

\tabincell{c}{\emph{replace}}
& \tabincell{c}{275}
& \tabincell{c}{{73}}
& \tabincell{c}{39}
& \tabincell{c}{6.17}
& \tabincell{c}{13.60}
& \tabincell{c}{211}
& \tabincell{c}{48}
& \tabincell{c}{128}
& \tabincell{c}{11.21}
& \tabincell{c}{10.84}
& \tabincell{c}{339}
& \tabincell{c}{48}
& \tabincell{c}{0}
& \tabincell{c}{101.47}
& \tabincell{c}{93.20}
\\ \hline

\tabincell{c}{\emph{totinfo}}
& \tabincell{c}{-}
& \tabincell{c}{-}
& \tabincell{c}{279}
& \tabincell{c}{-}
& \tabincell{c}{-}
& \tabincell{c}{76}
& \tabincell{c}{24}
& \tabincell{c}{179}
& \tabincell{c}{14.80}
& \tabincell{c}{11.50}
& \tabincell{c}{69}
& \tabincell{c}{{209}}
& \tabincell{c}{1}
& \tabincell{c}{54.36}
& \tabincell{c}{7.68}
\\ \hline

\tabincell{c}{\emph{printtokens}}
& \tabincell{c}{165}
& \tabincell{c}{57}
& \tabincell{c}{18}
& \tabincell{c}{6.15}
& \tabincell{c}{13.67}
& \tabincell{c}{178}
& \tabincell{c}{58}
& \tabincell{c}{4}
& \tabincell{c}{8.94}
& \tabincell{c}{6.22}
& \tabincell{c}{169}
& \tabincell{c}{{71}}
& \tabincell{c}{0}
& \tabincell{c}{15.94}
& \tabincell{c}{9.26}
\\ \hline

\tabincell{c}{\emph{printtokens2}}
& \tabincell{c}{188}
& \tabincell{c}{4}
& \tabincell{c}{0}
& \tabincell{c}{13.35}
& \tabincell{c}{7.25}
& \tabincell{c}{188}
& \tabincell{c}{4}
& \tabincell{c}{0}
& \tabincell{c}{13.21}
& \tabincell{c}{6.48}
& \tabincell{c}{187}
& \tabincell{c}{{5}}
& \tabincell{c}{0}
& \tabincell{c}{28.29}
& \tabincell{c}{28.89}
\\ \hline

\tabincell{c}{\emph{schedule}}
& \tabincell{c}{37}
& \tabincell{c}{0}
& \tabincell{c}{81}
& \tabincell{c}{0.05}
& \tabincell{c}{-}
& \tabincell{c}{92}
& \tabincell{c}{22}
& \tabincell{c}{4}
& \tabincell{c}{7.82}
& \tabincell{c}{11.13}
& \tabincell{c}{85}
& \tabincell{c}{{33}}
& \tabincell{c}{0}
& \tabincell{c}{33.04}
& \tabincell{c}{31.15}
\\ \hline

\tabincell{c}{\emph{schedule2}}
& \tabincell{c}{33}
& \tabincell{c}{0}
& \tabincell{c}{74}
& \tabincell{c}{0.04}
& \tabincell{c}{-}
& \tabincell{c}{42}
& \tabincell{c}{0}
& \tabincell{c}{65}
& \tabincell{c}{7.32}
& \tabincell{c}{-}
& \tabincell{c}{35}
& \tabincell{c}{{55}}
& \tabincell{c}{17}
& \tabincell{c}{189.03}
& \tabincell{c}{205.14}
\\ \hline

\tabincell{c}{\emph{cdaudio}}
& \tabincell{c}{544}
& \tabincell{c}{179}
& \tabincell{c}{50}
& \tabincell{c}{0.41}
& \tabincell{c}{0.81}
& \tabincell{c}{-}
& \tabincell{c}{190}
& \tabincell{c}{583}
& \tabincell{c}{-}
& \tabincell{c}{6.36}
& \tabincell{c}{566}
& \tabincell{c}{{207}}
& \tabincell{c}{0}
& \tabincell{c}{1.50}
& \tabincell{c}{1.58}
\\ \hline

\tabincell{c}{\emph{diskperf}}
& \tabincell{c}{270}
& \tabincell{c}{117}
& \tabincell{c}{56}
& \tabincell{c}{0.16}
& \tabincell{c}{0.41}
& \tabincell{c}{265}
& \tabincell{c}{119}
& \tabincell{c}{59}
& \tabincell{c}{5.08}
& \tabincell{c}{5.18}
& \tabincell{c}{304}
& \tabincell{c}{{139}}
& \tabincell{c}{0}
& \tabincell{c}{0.89}
& \tabincell{c}{0.85}
\\ \hline

\tabincell{c}{\emph{floppy}}
& \tabincell{c}{240}
& \tabincell{c}{69}
& \tabincell{c}{22}
& \tabincell{c}{0.18}
& \tabincell{c}{0.43}
& \tabincell{c}{244}
& \tabincell{c}{65}
& \tabincell{c}{22}
& \tabincell{c}{4.75}
& \tabincell{c}{5.23}
& \tabincell{c}{250}
& \tabincell{c}{{81}}
& \tabincell{c}{0}
& \tabincell{c}{0.72}
& \tabincell{c}{0.71}
\\ \hline

\tabincell{c}{\emph{floppy2}}
& \tabincell{c}{497}
& \tabincell{c}{82}
& \tabincell{c}{27}
& \tabincell{c}{0.33}
& \tabincell{c}{0.59}
& \tabincell{c}{501}
& \tabincell{c}{79}
& \tabincell{c}{26}
& \tabincell{c}{5.28}
& \tabincell{c}{5.68}
& \tabincell{c}{511}
& \tabincell{c}{{95}}
& \tabincell{c}{0}
& \tabincell{c}{1.51}
& \tabincell{c}{1.35}
\\ \hline

\tabincell{c}{\emph{kbfiltr}}
& \tabincell{c}{107}
& \tabincell{c}{49}
& \tabincell{c}{20}
& \tabincell{c}{0.09}
& \tabincell{c}{0.10}
& \tabincell{c}{107}
& \tabincell{c}{51}
& \tabincell{c}{18}
& \tabincell{c}{3.85}
& \tabincell{c}{3.61}
& \tabincell{c}{116}
& \tabincell{c}{{60}}
& \tabincell{c}{0}
& \tabincell{c}{0.32}
& \tabincell{c}{0.32}
\\ \hline

\tabincell{c}{\emph{kbfiltr2}}
& \tabincell{c}{249}
& \tabincell{c}{74}
& \tabincell{c}{39}
& \tabincell{c}{0.15}
& \tabincell{c}{0.20}
& \tabincell{c}{249}
& \tabincell{c}{76}
& \tabincell{c}{37}
& \tabincell{c}{4.14}
& \tabincell{c}{4.28}
& \tabincell{c}{264}
& \tabincell{c}{{98}}
& \tabincell{c}{0}
& \tabincell{c}{0.56}
& \tabincell{c}{0.54}
\\ \hline

\tabincell{c}{\emph{s3\_srvr\_1a}}
& \tabincell{c}{123}
& \tabincell{c}{295}
& \tabincell{c}{156}
& \tabincell{c}{2.69}
& \tabincell{c}{1.37}
& \tabincell{c}{123}
& \tabincell{c}{295}
& \tabincell{c}{156}
& \tabincell{c}{4.94}
& \tabincell{c}{4.13}
& \tabincell{c}{170}
& \tabincell{c}{{404}}
& \tabincell{c}{0}
& \tabincell{c}{0.69}
& \tabincell{c}{0.69}
\\ \hline

\tabincell{c}{\emph{s3\_srvr\_1b}}
& \tabincell{c}{43}
& \tabincell{c}{96}
& \tabincell{c}{0}
& \tabincell{c}{0.36}
& \tabincell{c}{0.80}
& \tabincell{c}{43}
& \tabincell{c}{96}
& \tabincell{c}{0}
& \tabincell{c}{3.31}
& \tabincell{c}{3.26}
& \tabincell{c}{43}
& \tabincell{c}{96}
& \tabincell{c}{0}
& \tabincell{c}{0.16}
& \tabincell{c}{0.16}
\\ \hline

\tabincell{c}{\emph{s3\_clnt}}
& \tabincell{c}{625}
& \tabincell{c}{969}
& \tabincell{c}{83}
& \tabincell{c}{14.62}
& \tabincell{c}{4.86}
& \tabincell{c}{661}
& \tabincell{c}{1012}
& \tabincell{c}{4}
& \tabincell{c}{9.72}
& \tabincell{c}{5.12}
& \tabincell{c}{665}
& \tabincell{c}{1012}
& \tabincell{c}{0}
& \tabincell{c}{39.62}
& \tabincell{c}{41.52}
\\ \hline

\tabincell{c}{\emph{s3\_clnt\_termination}}
& \tabincell{c}{540}
& \tabincell{c}{964}
& \tabincell{c}{91}
& \tabincell{c}{15.16}
& \tabincell{c}{4.35}
& \tabincell{c}{582}
& \tabincell{c}{1012}
& \tabincell{c}{1}
& \tabincell{c}{10.11}
& \tabincell{c}{5.42}
& \tabincell{c}{583}
& \tabincell{c}{1012}
& \tabincell{c}{0}
& \tabincell{c}{22.57}
& \tabincell{c}{24.02}
\\ \hline

\tabincell{c}{\emph{s3\_srvr\_2}}
& \tabincell{c}{418}
& \tabincell{c}{1034}
& \tabincell{c}{678}
& \tabincell{c}{3.50}
& \tabincell{c}{5.21}
& \tabincell{c}{698}
& \tabincell{c}{1344}
& \tabincell{c}{88}
& \tabincell{c}{11.00}
& \tabincell{c}{5.25}
& \tabincell{c}{704}
& \tabincell{c}{{1420}}
& \tabincell{c}{6}
& \tabincell{c}{102.85}
& \tabincell{c}{128.09}
\\ \hline

\tabincell{c}{\emph{s3\_srvr\_7}}
& \tabincell{c}{393}
& \tabincell{c}{1073}
& \tabincell{c}{794}
& \tabincell{c}{3.34}
& \tabincell{c}{4.78}
& \tabincell{c}{712}
& \tabincell{c}{1458}
& \tabincell{c}{90}
& \tabincell{c}{11.09}
& \tabincell{c}{5.43}
& \tabincell{c}{721}
& \tabincell{c}{{1538}}
& \tabincell{c}{1}
& \tabincell{c}{100.42}
& \tabincell{c}{124.45}
\\ \hline

\tabincell{c}{\emph{s3\_srvr\_8}}
& \tabincell{c}{425}
& \tabincell{c}{1183}
& \tabincell{c}{714}
& \tabincell{c}{3.98}
& \tabincell{c}{5.07}
& \tabincell{c}{701}
& \tabincell{c}{1529}
& \tabincell{c}{92}
& \tabincell{c}{10.70}
& \tabincell{c}{5.58}
& \tabincell{c}{706}
& \tabincell{c}{{1604}}
& \tabincell{c}{12}
& \tabincell{c}{107.31}
& \tabincell{c}{137.57}
\\ \hline

\tabincell{c}{\emph{s3\_srvr\_10}}
& \tabincell{c}{414}
& \tabincell{c}{1060}
& \tabincell{c}{726}
& \tabincell{c}{5.00}
& \tabincell{c}{32.16}
& \tabincell{c}{678}
& \tabincell{c}{1432}
& \tabincell{c}{90}
& \tabincell{c}{8.40}
& \tabincell{c}{4.45}
& \tabincell{c}{683}
& \tabincell{c}{{1517}}
& \tabincell{c}{0}
& \tabincell{c}{125.92}
& \tabincell{c}{111.44}
\\ \hline

\tabincell{c}{\emph{s3\_srvr\_12}}
& \tabincell{c}{388}
& \tabincell{c}{1611}
& \tabincell{c}{1126}
& \tabincell{c}{4.13}
& \tabincell{c}{7.00}
& \tabincell{c}{759}
& \tabincell{c}{2231}
& \tabincell{c}{135}
& \tabincell{c}{9.86}
& \tabincell{c}{6.19}
& \tabincell{c}{758}
& \tabincell{c}{{2345}}
& \tabincell{c}{22}
& \tabincell{c}{125.43}
& \tabincell{c}{144.04}
\\ \hline

\tabincell{c}{\emph{s3\_srvr\_13}}
& \tabincell{c}{431}
& \tabincell{c}{1111}
& \tabincell{c}{783}
& \tabincell{c}{4.43}
& \tabincell{c}{5.61}
& \tabincell{c}{745}
& \tabincell{c}{1500}
& \tabincell{c}{80}
& \tabincell{c}{10.04}
& \tabincell{c}{4.55}
& \tabincell{c}{737}
& \tabincell{c}{{1569}}
& \tabincell{c}{19}
& \tabincell{c}{111.75}
& \tabincell{c}{137.98}
\\ \hline

\tabincell{c}{\emph{osek\_control}}
& \tabincell{c}{607}
& \tabincell{c}{150}
& \tabincell{c}{170}
& \tabincell{c}{9.43}
& \tabincell{c}{8.09}
& \tabincell{c}{645}
& \tabincell{c}{199}
& \tabincell{c}{87}
& \tabincell{c}{7.72}
& \tabincell{c}{6.54}
& \tabincell{c}{623}
& \tabincell{c}{{277}}
& \tabincell{c}{27}
& \tabincell{c}{52.76}
& \tabincell{c}{65.12}
\\ \hline

\tabincell{c}{\emph{space\_control}}
& \tabincell{c}{1012}
& \tabincell{c}{457}
& \tabincell{c}{270}
& \tabincell{c}{13.34}
& \tabincell{c}{14.72}
& \tabincell{c}{1156}
& \tabincell{c}{495}
& \tabincell{c}{88}
& \tabincell{c}{9.85}
& \tabincell{c}{10.57}
& \tabincell{c}{1137}
& \tabincell{c}{{579}}
& \tabincell{c}{23}
& \tabincell{c}{67.23}
& \tabincell{c}{75.94}
\\ \hline

\tabincell{c}{\emph{subway\_control}}
& \tabincell{c}{1543}
& \tabincell{c}{842}
& \tabincell{c}{510}
& \tabincell{c}{21.52}
& \tabincell{c}{25.73}
& \tabincell{c}{1793}
& \tabincell{c}{1013}
& \tabincell{c}{89}
& \tabincell{c}{21.18}
& \tabincell{c}{14.12}
& \tabincell{c}{1787}
& \tabincell{c}{{1069}}
& \tabincell{c}{27}
& \tabincell{c}{93.91}
& \tabincell{c}{121.67}
\\ \hline
\hline

\tabincell{c}{\textbf{Total}}
& \tabincell{c}{9882}
& \tabincell{c}{11648}
& \tabincell{c}{6824}
& -
& -
& \tabincell{c}{11750}
& \tabincell{c}{14455}
& \tabincell{c}{2153}
& -
& -
& \tabincell{c}{12531}
& \tabincell{c}{15656}
& \tabincell{c}{155}
& -
& -
\\ \hline


\end{tabular}}
\end{table*}

Table~\ref{tab:model_checker_testing} gives the detailed performance statistics of the SMC-based reduction approach for data-flow testing, where "-" means the corresponding data does not apply or not available\footnote{BLAST hangs on \emph{totinfo}, and CPAchecker crashes on parts of pairs from \emph{cdaudio}.}.
For each implementation instance, it shows the number of feasible (denoted by $F$), infeasible (denoted by $I$) and unknown (denoted by $U$) pairs, and the median of testing times on feasible and infeasible pairs (denoted by $M_F$ and $M_I$, respectively).
The last row gives the total number of feasible, infeasible, and unknown pairs.
Note that BLAST and CPAchecker implement CEGAR-based model checking approach, thereby they can give the \emph{feasible} or \emph{infeasible} conclusion (or \emph{unknown} due to undecidability of the problem) without any false positives.
As for CBMC, it implements the bounded model checking technique, and in general cannot eliminate infeasible pairs as certain. Thus, the numbers of infeasible pairs identified by CBMC are only valid modulo the given checking bound.
From the results, we can see CPAchecker and CBMC are more effective than BLAST in terms of feasible pairs as well as infeasible pairs.
In detail, BLAST, CPAchecker and CBMC, respectively, cover 9882, 11750, 12531 feasible pairs, and identify 11648, 14455 and 15656 infeasible ones.

\begin{figure}[t]
	\begin{center}
		\includegraphics[width=0.7\textwidth]{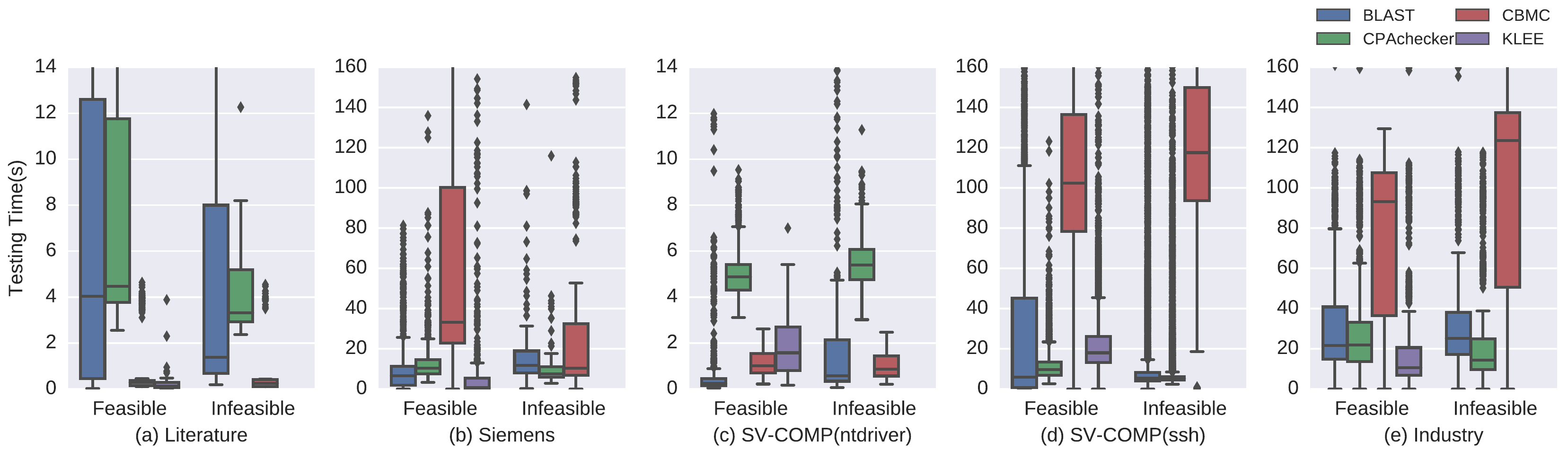}
	\end{center}
	\caption{Boxplot graphs of testing time for feasible and infeasible pairs concluded by BLAST, CPAchecker and CBMC (also including the testing time of KLEE configured with CPGS for feasible pairs) on the five subject groups, \ie, (a) Literature, (b) Siemens, (c) SV-COMP (ntd-driver), (d) SV-COMP (ssh), and (e) Industry.}
	\label{fig:mc_testing_time_feasible_infeasible}
\end{figure}

Fig.~\ref{fig:mc_testing_time_feasible_infeasible} shows the testing time of BLAST, CPAchecker and CBMC on each feasible and infeasible pair, respectively, categorized by the five benchmark groups, \ie, \emph{Literature}, \emph{Siemens}, \emph{SV-COMP (ntd-driver)}, \emph{SV-COMP (ssh)}, and \emph{Industry}.
On average, the cyclomatic complexity per function of them is 6.7, 4.7, 6.4, 88.5 and 10.2, respectively.
From the results of \emph{Literature} and \emph{SV-COMP (ntd-driver)} subjects in Fig.~\ref{fig:mc_testing_time_feasible_infeasible}, we can see CBMC is more efficient than BLAST and CPAchecker for both feasible and infeasible pairs.
The reason is that CBMC is an explicit path based model checking technique without abstraction, while CEGAR-based technique needs more additional time to setup the abstract program model before actual analysis, and takes more time during abstraction refinement.
Thus, CBMC is more suitable for small-scale subjects (\eg, the \emph{Literature} subjects) or medium-sized subjects without very complicated loop or recursion structures (\eg, the \emph{SV-COMP (ntd-driver)} subjects, but these subjects do have complicated call chains). 
For more complicated subjects, given a large checking bound (the loop bounds should be large enough to ensure as many feasible pairs as possible are covered), CBMC may degrade its performance in terms of testing time (see the column \emph{$M_I$*} of CBMC for the \emph{Siemens}, \emph{SV-COMP (ssh)} and \emph{Industry} subjects).
But we can also see CBMC actually covers 2649 and 781 more feasible pairs than BLAST and CPAchecker, respectively.
In contrast, the CEGAR-based approach is more efficient on \emph{Siemens}, \emph{SV-COMP (ssh)} and \emph{Industry} subjects, which are much more complex than \emph{Literature} and \emph{SV-COMP (ntd-driver)}.
The reason is that the CEGAR-based approach works on an initial program abstraction, and continuously refine the program states towards the target property (the test obligation in our context).  
Due to this model checking paradigm, the CEGAR-based approach could conclude infeasibility at a coarse program model level instead of the actual program level. In such cases, this CEGAR analysis greatly pays off, compared with the BMC-based approach.

Fig.~\ref{fig:venn_diagram_mc_tools} shows the venn diagrams of feasible, infeasible and unknown pairs concluded by the three model checkers BLAST, CPAchecker and CBMC.
We can get several important observations: (1) The number of feasible and infeasible pairs identified by all the three model checkers accounts for the majority, occupying 69.2\% and 71.9\% pairs, respectively. It indicates both the CEGAR-based and BMC-based approaches are practical and can give consistent answers in most cases. 
(2) Although the infeasible pairs identified by the BMC-based approach are only valid modulo the given checking bound, we can see CBMC in fact correctly concludes a large portion of infeasible pairs. Compared with the infeasiblity results of CPAchecker, 91.8\% (14,380/15,656) infeasible pairs identified by CBMC are indeed infeasible given appropriate checking bounds.
Thus, the BMC-based approach can still serve as a heuristic-criterion to identify hard-to-cover (probably infeasible) pairs, and better prioritize testing efforts.
(3) CPAchecker and CBMC have the largest number of overlapped pairs than the other combinations. They identify 94.7\% feasible and 90.3\% infeasible pairs, respectively. It indicates these two tools are more effective.

\begin{figure}
	\begin{center}
		\includegraphics[width=0.6\textwidth]{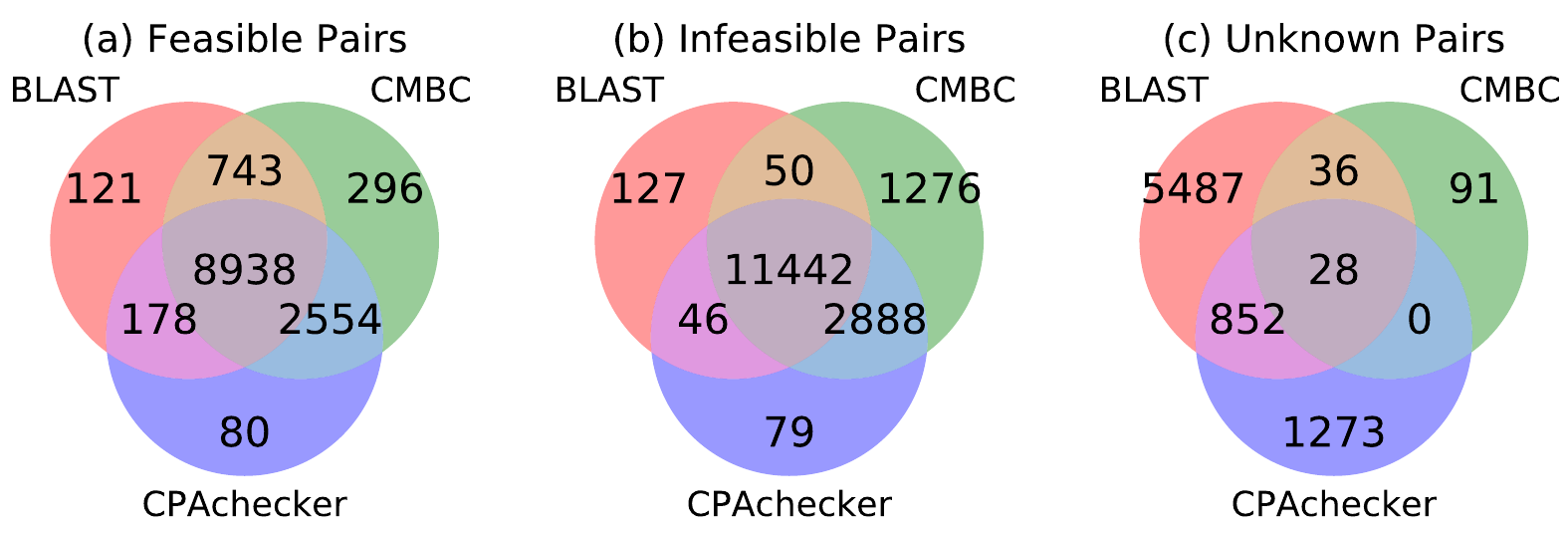}
	\end{center}
	\caption{Venn diagrams of (a) feasible, (b) infeasible and (c) unknown pairs concluded by the three model checkers BLAST, CPAchecker and CBMC for all subjects.}
	\label{fig:venn_diagram_mc_tools}
\end{figure} 

However, the performance of these model checkers still differs on different subjects in terms of testing time and the number of identified pairs.
An important reason is that their abstraction algorithms, implementation languages, underlying constraint solvers, search heuristics, built-in libraries have different impact on their performance. For example, BLAST, CPAchecker and CBMC are mainly implemented in OCaml, Java and C++, respectively; BLAST and CPAchecker respectively use CVC3 and Z3 as underlying constraint solvers.
Additionally, the time for setting up the program models before they can start the actual analysis process may also vary for different subjects or pairs, which should also be considered when interpreting these results.

\vspace{2pt}
\noindent\fbox{
	\parbox{0.95\linewidth}{
		\textbf{Answer to RQ2:} \textit{In summary, the SMC-based reduction approach is practical for data-flow testing. Both the CEGAR-based and BMC-based approaches can give consistent conclusions on the majority of def-use pairs. Specifically, the CEGAR-based approach can give answers for feasibility as certain, while the BMC-based approach can serve as a heuristic-criterion to identify hard-to-cover (probably infeasible) pairs when given appropriate checking bounds. In general, for data-flow testing, the CEGAR-based approach is more efficient on large and complicated programs, while the BMC-based approach is better for small/medium-sized programs. } 
	}
}

\begin{figure*}[t]
	\begin{center}
		\includegraphics[width=\textwidth]{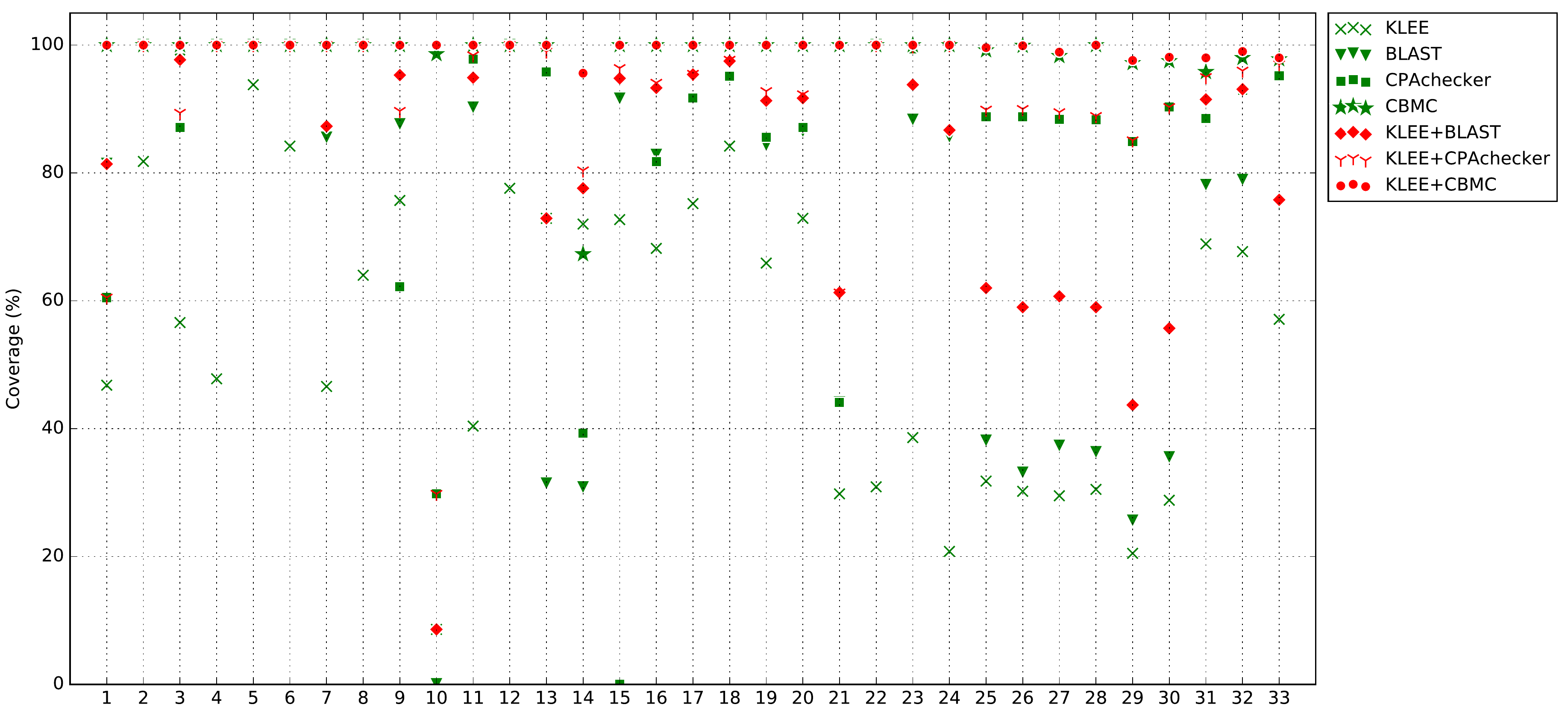}
	\end{center}
	\caption{Data-flow coverage achieved by KLEE, BLAST, CPAchecker, CBMC and their combinations (\ie, KLEE+BLAST, KLEE+CPAchecker, KLEE+CBMC) within the same time budget. Each number on the X axis denotes the set of 33 subjects in our study. Note that the results of CBMC and KLEE+CBMC are only valid modulo the given checking bounds.}
	\label{fig:klee_mc_combine_coverage}
\end{figure*}

\begin{figure*}[t]
	\begin{center}
		\includegraphics[width=\textwidth]{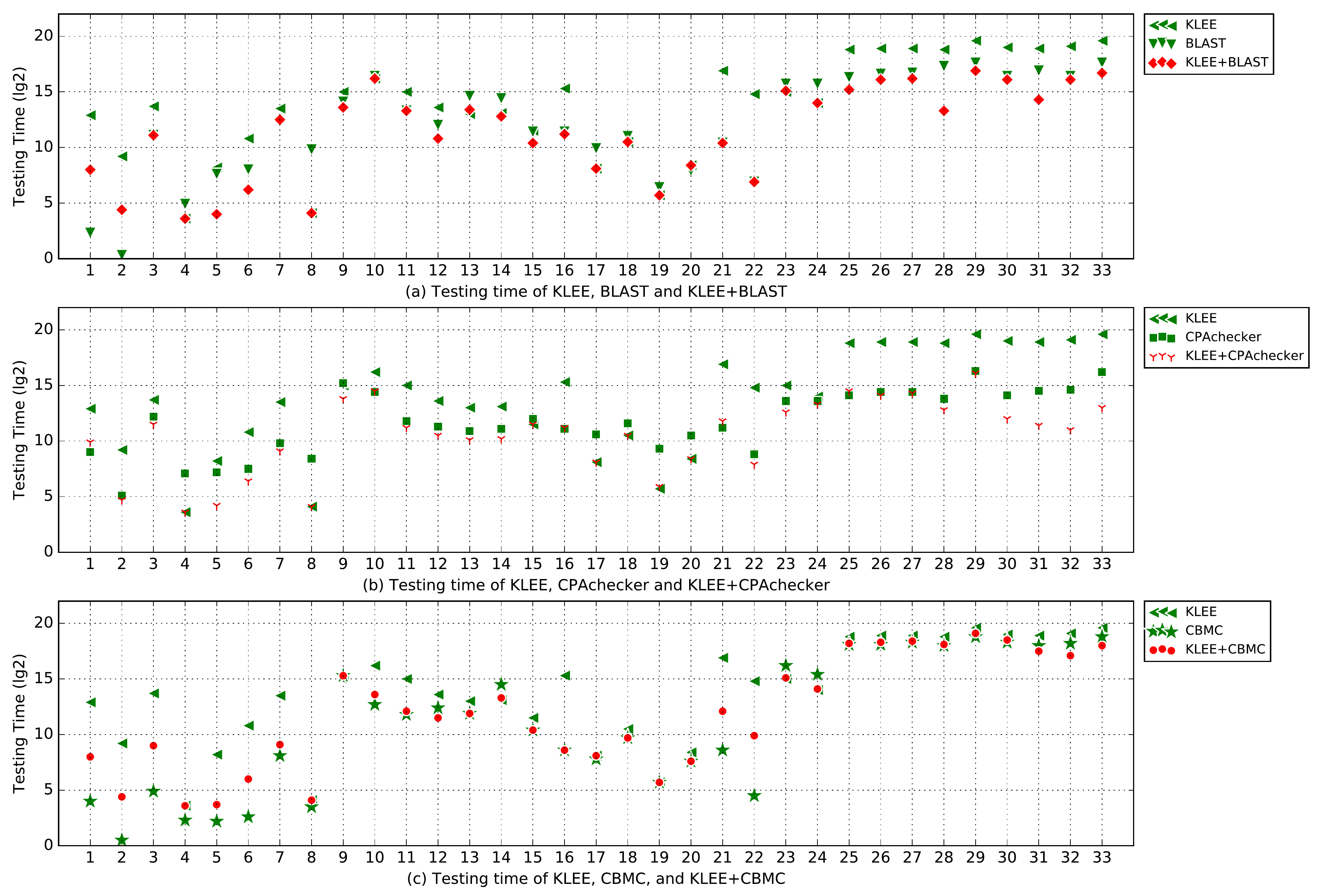}
	\end{center}
	\caption{Consumed time for data-flow testing by KLEE, BLAST, CPAchecker, CBMC and their combinations (\ie, KLEE+BLAST, KLEE+CPAchecker, KLEE+CBMC) for reaching their respective highest coverage. Each point on the X axis denotes the set of 33 subjects in our study. Note that the Y axis uses a logarithmic scale.}
	\label{fig:klee_mc_combination_testing_time}
\end{figure*} 

\subsection{Study 3}
\label{sec:study3}

To investigate the effectiveness of our combined approach, we complement the SE-based approach with the SMC-based approach to do data-flow testing. Specifically, we realize this combined approach by interleaving these two approaches (the setting is specified in Section~\ref{sec:rqs_and_studies}).
Fig.~\ref{fig:klee_mc_combine_coverage} shows the data-flow coverage achieved by KLEE, BLAST, CPAchecker, CBMC alone and their combinations (\eg, the combination of KLEE and CPAchecker, denoted as KLEE+CPAchecker for short) on the 33 subjects within the same testing budget. 
We can see the combined approach can greatly improve data-flow coverage. 
In detail, KLEE only achieves 54.3\% data-flow coverage on average for the 33 subjects, while KLEE+BLAST, KLEE+CPAchecker, and KLEE+CBMC, respectively, achieve 82.1\%, 90.8\%, and 99.5\% data-flow coverage.
Compared with KLEE, the combined approach instances, KLEE+BLAST, KLEE+CPAchecker, and KLEE+CBMC, respectively, improve the coverage by 27.8\%, 36.5\% and 45.2\% on average.
On the other hand, KLEE+BLAST improves coverage by 10\% against BLAST alone, and KLEE+CPAchecker improves coverage by 7\% against CPAchecker alone, respectively.

Fig.~\ref{fig:klee_mc_combination_testing_time} further shows the total testing time consumed by KLEE, BLAST, CPAchecker, CBMC and their combinations when achieving their peak coverage in Fig.~\ref{fig:klee_mc_combine_coverage}.
We can see that the combined approach can almost consistently reduce the total testing time on each subject.
Specifically, compared with KLEE, the combined approach instances, KLEE+BLAST, KLEE+CPAchecker, and KLEE+CBMC, respectively, achieve faster data-flow testing in 30/33, 29/33, and 28/33 subjects, and reduce the total testing time by 78.8\%, 93.6\% and 20.1\% on average in those subjects.
Among the three instances of combined approach, KLEE+CPAchecker achieves the best performance, which reduces testing time by 93.6\% for all the 33 subjects, and at the same time improves data-flow coverage by 36.5\%.
On the other hand, the combined approach instances,  KLEE+BLAST and KLEE+CPAchecker, also reduce the total testing time of BLAST and CPAchecker by 23.8\% and 19.9\%, respectively.

Based on the results of Fig.~\ref{fig:klee_mc_combine_coverage} and Fig.~\ref{fig:klee_mc_combination_testing_time}, we can conclude that the combined approach is more efficient. The reasons can be explained as follows.
Fig.~\ref{fig:klee_mc_feasible_pairs} shows the differences of the SE-based approach (KLEE) and the SMC-based approach (BLAST, CPAchecker and CBMC) in detecting feasible pairs. 
We can observe that the SE-based approach can cover a large portion of feasible pairs detected by the SMC-based approach.
Further, by comparing the testing time spent on feasible pairs between
the SE-based and the SMC-based approach (see Fig.~\ref{fig:mc_testing_time_feasible_infeasible}), we can see that the SE-based approach is in general very effective in covering feasible pairs.
A reasonable explanation is that the SE-based approach is a dynamic explicit path-based testing method (enhanced by the guided search strategy), which can quickly find an path to cover a target pair.
Meanwhile, the CEGAR-based approach is a static model checking-based testing method, which requires more analysis time for feasible pairs. CEGAR needs to first construct the abstract program model and then iteratively refine the model during each CEGAR analysis until it can give definite answers. 
On the other hand, it is easier for the CEGAR-based approach to identify infeasible pairs while the SE-based approach has to check all possible paths before confirming which pairs are infeasible (cost too much time).
This can be confirmed by Fig.~\ref{fig:klee_mc_infeasible_pairs}, where a large number of unknown pairs by KLEE can be concluded as infeasible by the SMC-based approach.
Specifically, BLAST and CPAchecker can weed out the infeasible pairs that KLEE cannot infer by 71.9\% and 89.7\%, respectively.
As for the BMC-based approach, CBMC* and CBMC** can respectively weed out 97.2\% and 88.9\% infeasible pairs. Note that here CBMC* denotes the valid results modulo the given checking bounds, while CBMC** denotes the results of actual infeasible pairs \wrt the results from CPAchecker. 
Therefore, it is beneficial to combine the strengths of symbolic execution and software model checking to achieve more efficient and practical data-flow testing.

\vspace*{2pt}
\noindent\emph{\textbf{Discussion}} \quad 
Since the CEGAR-based approach (\ie, BLAST and CPAchecker) is an unbounded software model checking technique, it can conclude infeasible pairs as certain. As for the BMC-based approach (\ie, CBMC), it in general cannot conclude infeasible pairs as certain. Thus, for the BMC-based approach, the results in Fig.~\ref{fig:klee_mc_combine_coverage} and Fig.~\ref{fig:klee_mc_combination_testing_time} are only valid modulo the checking bounds (this explains why the coverage of CBMC and KLEE+CBMC is even higher than the other approach instances). However, as Fig.~\ref{fig:venn_diagram_mc_tools} shows, compared with the results of CPAchecker, the large portion of infeasible pairs (91.8\%) identified by CBMC are actually indeed infeasible when appropriate checking bounds are given. Therefore, we believe the BMC-based approach can still be used as a heuristic-criterion to identify hard-to-cover (probably infeasible) pairs, especially considering testing budgets are usually limited in practice. In this case, testers can prioritize their efforts when our combined approach is realized on the SE-BMC-based approach.

\begin{figure}
	\centering
	\subfloat{\includegraphics[scale=.4]{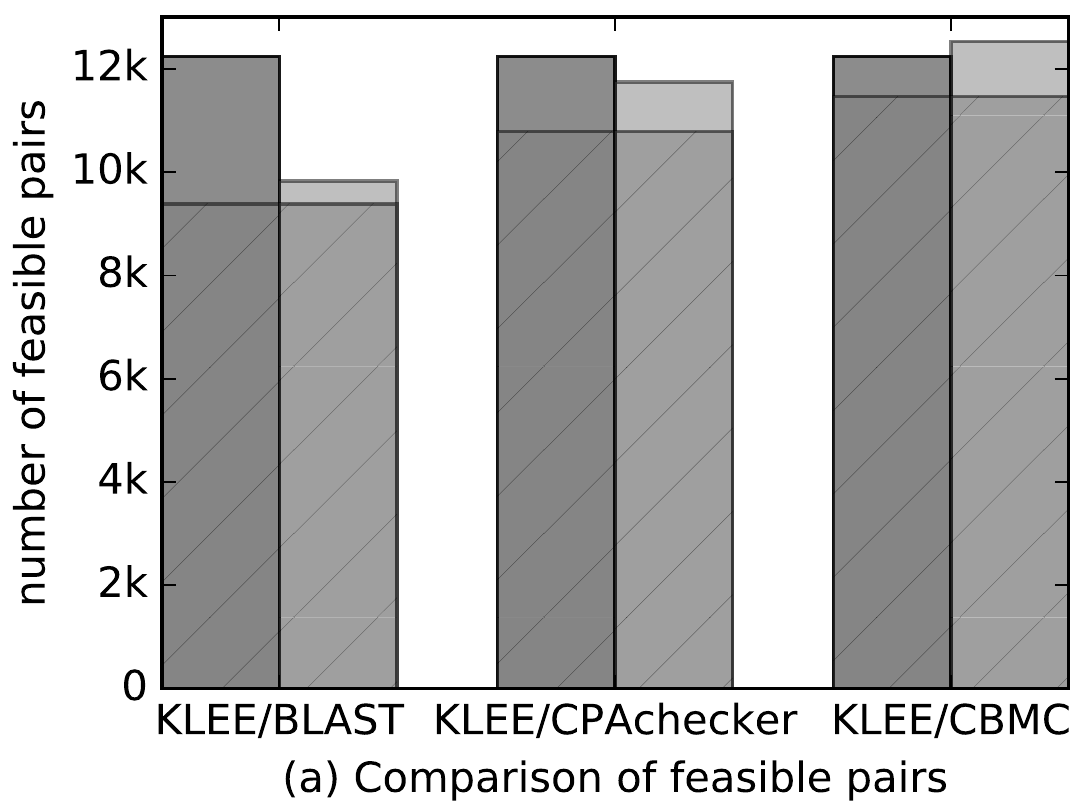}\label{fig:klee_mc_feasible_pairs}}
	\subfloat{\includegraphics[scale=.4]{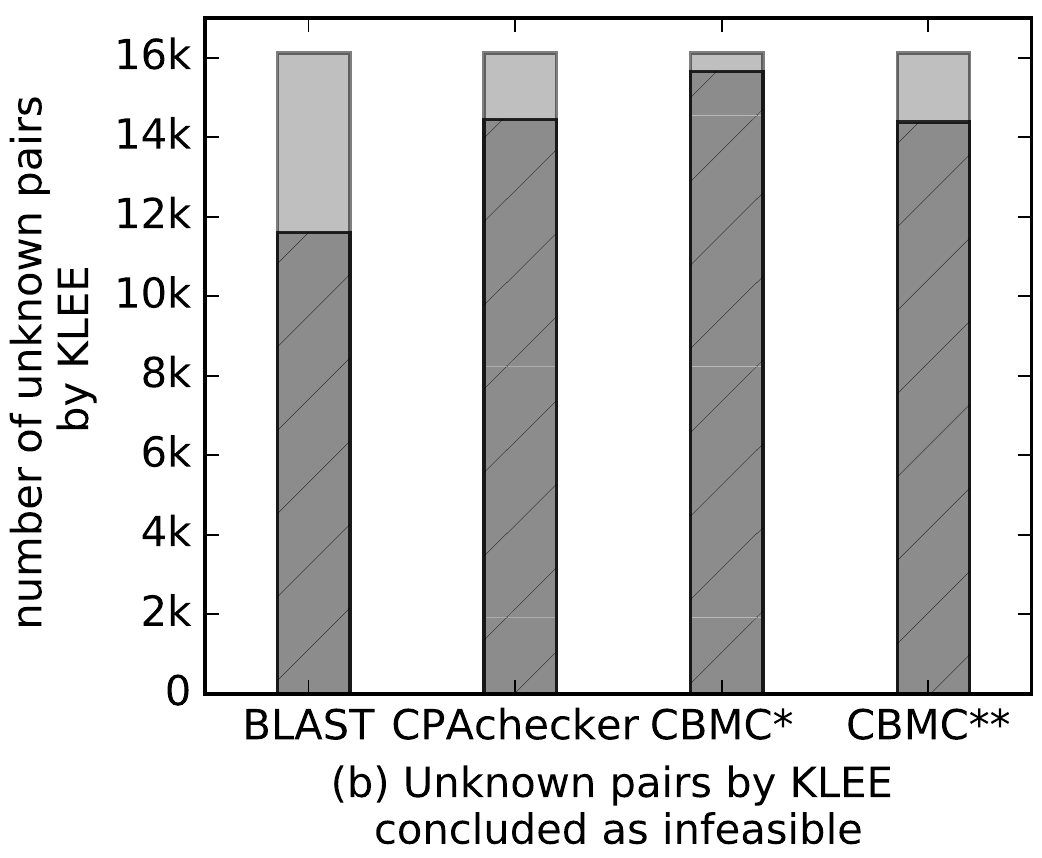}\label{fig:klee_mc_infeasible_pairs}}
	\caption{(a) Comparison of feasible pairs detected by KLEE, BLAST, CPAchecker and CBMC. The shadow parts represent those pairs that are covered by both KLEE and the corresponding model checker. (b) Unknown pairs of KLEE with the shadow parts representing those pairs that are concluded as infeasible by BLAST, CPAchecker and CBMC. CBMC* denotes the valid results modulo the given checking bounds, while CBMC** denotes the results of actual infeasible pairs \wrt the results from CPAchecker.}
	\label{fig:klee_mc_pairs}
\end{figure}

\vspace{2pt}
\noindent\fbox{
	\parbox{0.95\linewidth}{
		\textbf{Answer to RQ3:} \textit{In summary, the combined approach, which combines symbolic execution and software model checking, achieves more efficient data-flow testing. The model checking approach can weed out infeasible pairs that the symbolic execution approach cannot infer by 71.9\%$\sim$97.2\%. Compared with the SE-based approach alone, the combined approach can improve data-flow coverage by 27.8$\sim$45.2\%. In particular, the instance KLEE+CPAchecker performs best, which reduces total testing time by 93.6\% for all 33 subjects, and at the same time improves data-flow coverage by 36.5\%. Compared with the CEGAR-based or BMC-based approach alone, the combined approach can also reduce testing time by 19.9$\sim$23.8\%, and improve data-flow coverage by 7$\sim$10\%.} 
	}
}

\section{Discussion}
\label{sec:discussion}

This section gives the detailed discussion about (1) limitations of our approach; (2) our experience of data-flow testing and the applications for other testing scenarios, and (3) threats to validity.

\subsection{Limitations}
\label{sec:limitation}

\noindent\textbf{\emph{Undecidability of data-flow based test generation}}. \quad
The problem of data-flow based test generation can be formulated as a reachability checking problem~\cite{MaKFH11,GodefroidLM08}, \ie, whether we can find a particular execution path (or a particular program state) of the program under test to cover a given def-use pair. This problem itself is undecidable, and in practice our testing technique cannot always answer the feasibility of a target pair within a given time budget. Specifically, the symbolic execution technique may fail to generate a test case for an actual feasible pair, since it cannot exhaustively explore the whole (probably infinite) program state space (no matter which search strategies are adopted). As for software model checking, the CEGAR-based approach may answer \emph{unknown} when it cannot terminate within a given time bound, and the infeasibility answers of BMC-based approach are valid modulo the given checking bounds. But combining symbolic execution and model checking can mitigate this issue to some extent. As demonstrated by Fig.~\ref{fig:klee_mc_pairs} in the evaluation, if one technique cannot cover a given pair, the other one may be able to cover it.

\vspace*{3pt}
\noindent\textbf{\emph{Technical Limitations}}. \quad
Our hybrid data-flow testing framework is built on top of KLEE symbolic execution engine and different C software model checkers. Therefore, our framewok inherited all technical limitations or implementation issues of these tools. These limitations or issues may lead to inconsistent testing results. For example, if KLEE concludes a given pair is feasible with a test input, while a software model checker concludes this pair is infeasible, an inconsistent case occurs. 

To investigate this issue, we cross-checked the testing results between the symbolic execution component (KLEE) and the software model checking component (CPAchecker and CBMC) on all 28,354 def-use pairs from the 33 program subjects. 
We find our testing framework in general is highly reliable and robust --- it gives consistent testing results for almost all def-use pairs (99.76\% def-use pairs). Only 68 cases out of 28,354 pairs are inconsistent. We carefully inspected these inconsistent cases, and found all of them were caused by some technical limitations of KLEE or the model checkers. Our testing algorithm itself is correct. 

For example, KLEE may output a test input for an actual infeasible def-use pair if the interface \texttt{klee\_make\_symbolic} for symbolizing input variables is used in program loops.
This usage disturbs the predecessor-successor relations between \texttt{ProcessTree} nodes (one node corresponds to a block of sequential instructions in KLEE), which are used to decide the feasibility of a pair. CPAchecker or CBMC may give incorrect testing results if the program under test uses some not well-supported C standard library functions.
Also, the applicability of our technique is constrained by the language features or data structures supported by KLEE and software model checkers. 
But we believe these limitations will be overcome by the continuous advances in both symbolic execution and software model checking community. For example, 
the ability and reliability of the implementations of these techniques are improving~\cite{HVC17-testing,kapus2017automatic,differential_model_checker}, and software model checking is now driven by active research (\eg, the annual SV-COMP competition~\cite{SV-COMP}) and industrial adoption~\cite{BallR02,PavlinovicLS16,KhoroshilovMPZ09}.

\subsection{Experience of data-flow testing and Other Applications}

\vspace*{3pt}
\noindent\textbf{\emph{Experience with concolic testing}}. \quad 
We have implemented our symbolic execution-based data-flow testing approach on two variant techniques~\cite{CadarS13}, \ie, execution-generated testing (adopted by KLEE~\cite{KLEE}) and concolic testing (adopted by DART~\cite{DART}, CUTE~\cite{CUTEC}, CREST~\cite{CREST} and CAUT~\cite{SuPFHYJZ14}). 
In addition to KLEE, we also applied the widely-used concolic testing tool CREST developed by Burnim \etal~\cite{CREST} and our own tool CAUT on these benchmarks.
Different from execution-generated testing, concolic testing needs initial concrete values for input variables, and maintains the entire concrete state of program along each execution path.
Due to this feature, we find CREST and CAUT can achieve faster data-flow testing than KLEE on some subjects (\eg, some SIR programs). 
However, they cannot work well on complicated subjects due to their technical limitations. For example, they do not well-support complicate data structures (\eg, \texttt{pointer}, \texttt{struct}, \texttt{list}), and thus lead to low data-flow coverage or even tool crashes. 
This also confirms the benefit of implementing our approach on the execution-generated testing tool KLEE.

\vspace*{3pt}
\noindent\textbf{\emph{Other Applications}}. \quad 
Currently, our testing technique is designed for classic data-flow testing. Given a def-use pair, it finds a test input for a given pair at one time. This technique can also benefit directed testing or debugging scenarios~\cite{XieTHS09,ZamfirC10,MaKFH11,MarinescuC13}. If targeting all def-use pairs at once, our technique, especially the SE-based approach, can be further enhanced by those generic search strategies (\eg, RSS-MD2U) and collateral coverage-based techniques~\cite{Marre96,Marre03}. 
Our testing framework could also be used to achieve other coverage criteria~\cite{SuPMHS16}, and investigate the effectiveness of different criteria~\cite{HutchinsFGO94,HassanA13,Inozemtseva14}.

\subsection{Threats to Validity}

\vspace*{3pt}
\noindent\textbf{\emph{Internal Validity}}. \quad 
Our results may be affected by the comprehensiveness and correctness of the evaluation.
For example, the testing results may be affected by the different characteristics and configurations of the tools (\eg, implementation languages, constraint solvers, search heuristics, abstraction algorithms, boundedness, the time for setting up program models before actual analysis) on which our technique was implemented.
To mitigate this issue, we (1) used the default configurations of these tools without any particular tuning from the perspective of end users, and (2) extensively evaluated our technique on collectively 28,354 pairs by four different tools (KLEE, BLAST, CPAchecker, CBMC) and five different search strategies to gain more overall understanding.
Specifically, these four tools cover same technique with different implementations (BLAST \vs CPAchecker), same approach category achieved by different techniques (CPAchecker \vs CBMC), and different approaches (KLEE \vs CPAchecker). 
All evaluations were repeated 30 times to mitigate algorithm randomness, and the results were examined by statistical significant test.
Additionally, in our current implementation, we do not identify def-use pairs induced by pointer aliases in the risk of missing some test objectives. More sophisticated data-flow analysis techniques (\eg, dynamic data-flow analysis~\cite{DenaroPV14}) or tools (\eg, Frama-C~\cite{KirchnerKPSY15}, SVF~\cite{SuiX16}) can be used to mitigate this problem. However, we believe that this is an independent issue and not the focus of this work. The effectiveness of our SE-based and MC-based approach should remain.

\vspace*{3pt}
\noindent\textbf{\emph{External Validity}}. \quad 
Our results may be affected by the representativeness of the subjects. To mitigate this issue, we dedicatedly constructed a benchmark repository of 33 subjects. Although some benchmarks are relatively small-sized, we chose them with these considerations: (1) they were from different sources, \ie, prior data-flow testing research work, standard benchmarks for program analysis and testing, standard benchmarks for program verification, and real-world industrial projects; (2) they contain diverse data-flow scenarios, including mathematical computations, standard algorithms, utility programs, OS device drivers, and embedded control software; and (3) they were carefully checked to reduce evaluation bias and ensure all tools can correctly reason them (\eg, by adding necessary function stubs).
Therefore, we believe they contain the typical characteristics of real-world programs, and are also fair to all tools. 
From these benchmarks, the effectiveness of our approach is evident. Although it is interesting to consider more program subjects, due to our novel, general methodologies, we believe that the results should be consistent.

\vspace*{3pt}
\noindent\textbf{\emph{Construct Validity}}. \quad 
As demonstrated by the cross-checking results in Section~\ref{sec:limitation}, the implementation of our testing framework is highly reliable and robust. It gives consistent testing results for almost all def-use pairs. With the advances in symbolic execution and software model checking, our testing framework will be further enhanced.
\section{Related Work}
\label{sec:related}
This section discusses three strands of closely related work: (1) data-flow based test generation, (2) directed symbolic execution, and (3) infeasible test objective detection.

\subsection{Data-flow based test generation}

Data-flow testing has been continuously investigated in the past four decades~\cite{Frankl93,ForemanZ93,Weyuker93,HutchinsFGO94,FranklI98}. Existing work can be categorized into five main categories according to the testing techniques. We only discuss typical literature work here. Readers can refer to a recent survey~\cite{susurvey2015} for details. 

The most widely used approach to is \emph{search-based testing}, which utilizes meta-heuristic search techniques to identify test inputs for target def-use pairs.
Girgis~\cite{Girgis05} first uses Genetic Algorithm (GA) for Fortran programs, and Ghiduk \etal~\cite{GhidukHG07} use GA for C++ programs.
Later, Vivanti \etal~\cite{VivantiMGF13} and Denaro \etal~\cite{DenaroMPV15} apply GA to \texttt{Java} programs by the tool EvoSuite.
Some optimization-based search algorithms~\cite{NayakM10,Singla11,Singla11a,Ghiduk10b} are also used, but they have only evaluated on small programs without available tools. \emph{Random testing} is a baseline approach for data-flow testing~\cite{Girgis05,GhidukHG07,AlexanderOS10,DenaroMPV15,GirgisS14}. 
Some researchers use \emph{collateral coverage-based testing}~\cite{HarmanKLMY10}, which exploits the observation that the test case that satisfies one target test objective can also ``accidentally" cover the others.
Malevris~\etal~\cite{Malevris05} use branch coverage to achieve data-flow coverage. 
Merlo~\etal~\cite{MerloA99} exploit the coverage implication between data-flow coverage and statement coverage to achieve intra-procedural data-flow testing. Other efforts include~\cite{Marre96,Marre03,SantelicesSH06,Santelices07}.
Some researchers use \emph{traditional symbolic execution}. For example, Girgis~\cite{Girgis93} develops a simple symbolic execution system for DFT, which statically generates program paths \wrt a certain control-flow criterion (\eg, branch coverage), and then selects those executable ones that can cover the def-use pairs of interest.
Buy~\etal~\cite{BuyOP00} adopts three techniques, \ie, data-flow analysis, symbolic execution and 
automated deduction to perform data-flow testing.
However, they have provided little evidence of practice.
Hong~\etal~\cite{HongCLSU03} adopt \emph{classic CTL-based model checking} to generate data-flow test data. Specifically, the program is modeled as 
a Kripke structure and the requirements of data-flow coverage are characterized as a set of CTL property formulas. 
However, this approach requires manual intervention, and its scalability is also unclear.

Despite the plenty of work on data-flow based testing, they are either inefficient or imprecise. Our work is the first one to leverage the modern symbolic execution and software model checking techniques to achieve DFT efficiently and precisely.

\subsection{Directed Symbolic Execution} 

Much research~\cite{MaKFH11,MarinescuC13,XieTHS09,ZamfirC10,DoFP12} has been done to guide path search toward a specified program location via symbolic execution.
Do \etal~\cite{DoFP12} leverage data dependency analysis to guide the search to reach a particular program location, while we use dominator analysis.  Ma~\etal~\cite{MaKFH11} suggest a call chain backward search heuristic to find a feasible path, backward from the target program location to the entry.  However, it is difficult to adapt this approach on data-flow testing, because it requires that a function can be decomposed into logical parts when the target locations (\eg the \emph{def} and the \emph{use}) are located in the same function. But  decomposing a function itself is a nontrivial task.  Zamfir~\etal~\cite{ZamfirC10} narrow the path search space by following a \emph{limited} set of critical edges and a statically-necessary
combination of intermediate goals.  On the other hand, our approach finds a set of cut points from the program entry to the target locations, which makes path exploration more efficient.  Xie
\etal~\cite{XieTHS09} integrate fitness-guided path search strategy with other heuristics to reach a program point.  The proposed strategy is only efficient for those problems amenable to its fitness functions.  Marinescu \etal~\cite{MarinescuC13} use a shortest distance-based guided search method (like the adapted SDGS heuristic in our evaluation) with other heuristics to quickly reach the line of interest in patch testing. In contrast, we combine several search
heuristics to guide the path exploration to traverse two specified program locations (\ie the \emph{def} and \emph{use}) sequentially for data flow testing.

\subsection{Detecting Infeasible Test Objectives}

As for detecting infeasible test objectives, early work uses constraint-based technique~\cite{GoldbergWZ94,OffuttP97}. Offutt and Pan~\etal~\cite{OffuttP97} extract a set of path constraints that encode the test objectives from the program under test. Infeasible test objectives can be identified if the constraints do not have solutions.
Recent work by Beckman~\etal~\cite{BeckmanNRSTT10}, Baluda~\etal~\cite{BaludaBDP10,BaludaBDP11,BaludaDP16}, Bardin~\etal~\cite{BardinDDKPTM15} use weakeast precondition to identify infeasible statements and branches. For example, Baluda~\etal use model refinement with weakest precondition to exclude infeasible branches; Bardin~\etal applies weakest precondition with abstract interpretation to eliminate infeasible objectives.
Marcozzi~\etal~\cite{marcozzi2018time} also use weakest precondition to identify polluting test objectives (including infeasible, duplicate and subsumed) for condition, MC/DC and weak mutation coverage.
In contrast, our testing framework mainly use the CEGAR-based model checking technique to identify infeasible def-use pairs for data-flow testing.
One close work is from Daca~\etal~\cite{DacaGH16}, who combine concolic testing (CREST) and model checking (CPAchecker) to find a test suite \wrt branch coverage. Our work has some distinct differences with theirs. First, they target at branch coverage, while we enforce data-flow testing. Second, they directly modify the existing generic path search strategies of CREST, and backtrack the search if the explored direction has been proved as infeasible by CPAchecker. As a result, the performance of their approach (\ie, avoid unnecessary path explorations) may vary across different search strategies due to the paths are selected in different orders. 
In contrast, we implement a designated search strategy to guide symbolic execution, and realize the reduction approach directly on model checkers.
Although our approach is simple, it can treat model checkers as black-box tools without any modification and seamlessly integrate with KLEE.
Model checking techniques have recently been adapted to aid software testing~\cite{FraserWA09,HVC17-testing}. For example, FShell~\cite{HolzerTVS10} is another model checking based test generator, which uses CBMC as the basis and uses FQL for specification of coverage criteria. ESBMC-INCR~\cite{MorseRCN014} and ESBMC-KIND~\cite{GadelhaIC17} are the two enhanced variants of CBMC, which can also be used for test generation. ESBMC-INCR uses an iterative strategy to increase loop bounds and ESBMC-KIND computes loop invariants to improve performance. In our context, we use CBMC as a heuristic-criterion to detect hard-to-cover (probably infeasible) pairs.

\section{Conclusion}
\label{sec:conclusion}
This paper introduces an efficient, combined data-flow testing approach. We designed a cut point guided search strategy to make symbolic execution practical; and devised a simple encoding of data-flow testing via software model checking. 
The two approaches offer complementary strengths: SE is more effective at covering feasible def-use pairs, while SMC is more effective at rejecting infeasible pairs.
Specifically, the CEGAR-based approach is used to eliminate infeasible pairs as certain, while the BMC-based approach can be used as a heuristic-criterion to identify hard-to-cover (probably infeasible) pairs when given appropriate checking bounds (especially suitable for small/medium-sized programs without complicated loops).
The empirical evaluation results have demonstrated that our combined testing approach can reduce testing time by 20.1$\sim$93.6\% and improve data-flow coverage by 27.8\%$\sim$45.2\% than the enhanced SE-based approach alone; also reduce testing time by 19.9$\sim$23.8\%, and improve data-flow coverage by 7\%$\sim$10\% than the CEGAR-based/BMC-based approach alone.
This work not only provides novel techniques for data-flow testing, but also suggests a new perspective on this problem to benefit from advances in symbolic execution and model checking.

%

%
\bibliographystyle{ACM-Reference-Format}
\bibliography{sample-base}

\end{document}